\documentclass[acmsmall,svgnames,authorversion]{acmart}
\acmVolume{}
\acmArticle{}
\acmYear{}
\acmMonth{10}
\acmDOI{} 
\startPage{1}

\usepackage{hyperref}
\usepackage[utf8]{inputenc}
\usepackage[T1]{fontenc}
\usepackage{todonotes}
\usepackage{thmtools}
\usepackage{thm-restate}
\usepackage{defs}
\usepackage{proof}
\usepackage{float}
\usepackage{wrapfig}

\usepackage{stmaryrd}
 
\usepackage{xspace}
\usepackage{listings}

\lstdefinelanguage{ML}{
  alsoletter={*},
  morekeywords={datatype, of, if, *},
  sensitive=true,
  morecomment=[s]{/*}{*/},
  morestring=[b]"
}

\lstdefinelanguage{scala}{
  alsoletter={@},
  morekeywords={abstract, case, catch, choose, class, def, do, else, extends, final, finally, for, if, implicit, import, match, new, null, object, let, in, be, lazy,holds,@erasable,
override, package, private, protected, requires, ensures, decreases, return, sealed, super, then, this, throw, trait, try, type, val, var, while, yield, domain, template, assume, fun,
postcondition, precondition,invariant, constraint, assert, each, _, return, @generator, ensure, require, ensuring, assuming, otherwise, asserting}
  sensitive=true,
  morecomment=[l]{//},
  morecomment=[s]{/*}{*/},
  morestring=[b]"
}

\newcommand{\codestyle}{\small\sffamily}

\newcommand{\RA}{\Rightarrow}

\usepackage{booktabs}   
\usepackage{subcaption} 

\usepackage{mathtools} 
\usepackage{amssymb}

\setcopyright{none}        

\newtoggle{arxiv}

\toggletrue{arxiv}

\title{System FR as Foundations for Stainless}

\author{Jad Hamza}
\affiliation{
  \department{LARA}              
  \institution{EPFL}            
  \country{Switzerland}                    
}
\email{jad.hamza@epfl.ch}          

\author{Nicolas Voirol}
\affiliation{
  \department{LARA}              
  \institution{EPFL}            
  \country{Switzerland}                    
}
\email{nicolas.voirol@epfl.ch}          

\author{Viktor Kunčak}
\affiliation{
  \department{LARA}              
  \institution{EPFL}            
  \country{Switzerland}                    
}
\email{viktor.kuncak@epfl.ch}          

\begin{abstract}
We present the design, implementation, and foundation of a verifier
for higher-order functional programs with generics and recursive
data types. Our system supports proving safety and termination
using preconditions, postconditions and assertions. It supports writing proof hints
using assertions and recursive calls. To formalize the soundness of the system we introduce
System FR, a calculus  supporting System F polymorphism, dependent refinement
types, and recursive types (including recursion through contravariant positions of
function types). Through the
use of sized types, System FR supports reasoning about termination of lazy data
structures such as streams. We formalize a reducibility argument using the Coq
proof assistant and prove the soundness of a type-checker with
respect to call-by-value semantics, ensuring type safety and normalization for
typeable programs.
Our program verifier is implemented as an alternative verification-condition
generator for the Stainless tool, which relies on the Inox SMT-based
solver backend for automation.
We demonstrate the efficiency of our approach by verifying a
collection of higher-order functional programs comprising around 14000 lines of
polymorphic higher-order Scala code, including graph search algorithms, basic
number theory, monad laws, functional data structures, and assignments from
popular Functional Programming MOOCs.

\end{abstract}

\begin{document}
\sloppy
\maketitle



\newcounter{counttodos}
\newcommand{\dtodo}[1]{
  \stepcounter{counttodos}
  \todo[inline,color=blue!10]{
    \textcolor{blue!70!black}{TODO {\thecounttodos}: #1}
  }
}
\newcommand{\itodo}[1]{
  \stepcounter{counttodos}
  \todo[inline,color=red!10]{
    \textcolor{red!70!black}{TODO {\thecounttodos}: #1}
  }
}
\newcommand{\ptodo}[1]{
  \stepcounter{counttodos}
  \todo[inline,color=green!10]{
    \textcolor{green!30!black}{TODO {\thecounttodos}: #1}
  }
}
\newcommand{\ttodo}[1]{
  \stepcounter{counttodos}
  \todo[inline,color=orange!10]{
    \textcolor{orange!75!black}{TODO {\thecounttodos}: #1}
  }
}

\newcommand{\set}[1]{\{#1\}}
\newcommand{\Nat}{\mathbb{N}}
\newcommand{\define}{\triangleq}

\newcommand{\Var}{\mathcal{V}}
\newcommand{\Op}{{\sf Op}}
\newcommand{\wfset}{A}
\newcommand{\Fid}{{\sf Fid}}
\newcommand{\Term}[1]{{\sf Term}_{#1}}
\newcommand{\UTerms}{{\sf Terms}}
\newcommand{\Type}{{\sf Type}}
\newcommand{\Val}{{\sf Val}}

\newcommand{\tlambda}[3]{\lambda{#1}: {#2}.\ {#3}}
\newcommand{\tulambda}[2]{\lambda{#1}.\ {#2}}
\newcommand{\tapp}[2]{{#1}\ {#2}}
\newcommand{\tulet}[3]{{\tt let}\ {#1} = {#2}\ {\tt in}\ {#3}}
\newcommand{\tassert}[2]{{\tt assert }({#1})\ {\tt in}\ {#2}}
\newcommand{\ttrue}{{\tt true}}
\newcommand{\tfalse}{{\tt false}}
\newcommand{\tfold}[2]{{\tt fold}[#1]({#2})}
\newcommand{\tunfoldin}[2]{{\tt unfold}\ {#1}\ {\tt in}\ {#2}}
\newcommand{\tunfoldposin}[2]{{\tt unfold\_pos}\ {#1}\ {\tt in}\ {#2}}
\newcommand{\tufold}[1]{{\tt fold}({#1})}
\newcommand{\tu}{{\tt ()}}
\newcommand{\terr}[1]{{\tt err}[{#1}]}
\newcommand{\tuerr}{{\tt err}}
\newcommand{\tprojl}[1]{\tapp{\pi_1}{#1}}
\newcommand{\tprojr}[1]{\tapp{\pi_2}{#1}}
\newcommand{\tite}[3]{{\tt if}\ {#1}\ {\tt then}\ {#2}\ {\tt else}\ {#3}}
\newcommand{\ttite}[3]{{\tt ifNonempty}\ {#1}\ {\tt then}\ {#2}\ {\tt else}\ {#3}}
\newcommand{\wflt}{\prec}
\newcommand{\wfleq}{\preceq}
\newcommand{\wftop}{\top}
\newcommand{\wfltop}{{\tt lt}}
\newcommand{\tabs}[2]{\Lambda{#1}.\ {#2}}
\newcommand{\tinst}[2]{{#1}[{#2}]}
\newcommand{\tuabs}[1]{\Lambda{#1}}
\newcommand{\tuinst}[1]{\tinst{#1}{}}
\newcommand{\tinstforall}[2]{{\tt inst}({#1},{#2})}
\newcommand{\tsize}[1]{{\tt size}({#1})}
\newcommand{\sizesemantics}[1]{{\tt size\_sem}({#1})}
\newcommand{\tujudge}[2]{{#1}\, {:}\hspace{-0.08em}{:}\hspace{-0.08em}\, {#2}}
\newcommand{\tsetcomp}[1]{{#1}^{\textsf c}}
\newcommand{\timage}[2]{{#1}[{#2}]}
\newcommand{\tinvimage}[2]{{#1}^{-1}[{#2}]}

\newcommand{\tzero}{{\tt zero}}
\newcommand{\tsucc}[1]{{\tt succ}({#1})}
\newcommand{\trec}[4]{{\tt rec}[{#1}]({#2},{#3},{#4})}
\newcommand{\turec}[3]{{\tt rec}({#1},{#2},{#3})}
\newcommand{\tmatch}[3]{{\tt match}({#1},{#2},{#3})}
\newcommand{\tumatch}[3]{{\tt match}({#1},{#2},{#3})}
\newcommand{\tuleft}[1]{{\tt left}({#1})}
\newcommand{\turight}[1]{{\tt right}({#1})}
\newcommand{\tleft}[1]{{\tt left}({#1})}
\newcommand{\tright}[1]{{\tt right}({#1})}
\newcommand{\taleft}[2]{{\tt left}[{#1}]({#2})}
\newcommand{\taright}[2]{{\tt right}[{#1}]({#2})}
\newcommand{\tsummatch}[3]{{\tt either\_match}({#1},{#2},{#3})}
\newcommand{\tfix}[2]{{\tt fix}[{#1}]({#2})}
\newcommand{\tufix}[1]{{\tt fix}({#1})}
\newcommand{\predsymb}{{\tt pred}}
\newcommand{\tpred}[1]{\predsymb({#1})}
\newcommand{\turefl}{{\tt refl}}
\newcommand{\trefl}[2]{{\tt refl}[{#1},{#2}]}

\newcommand{\tunit}{{\sf Unit}}
\newcommand{\tbool}{{\sf Bool}}
\newcommand{\tbigint}{\mathbb{Z}}
\newcommand{\tnat}{{\sf Nat}}
\newcommand{\tsingleton}[1]{\{{#1}\}}
\newcommand{\trefine}[3]{\{{#1}:~{#2}~|~{#3}\}}
\newcommand{\dobracket}{{\{}\hspace*{-0.2em}{\{}}
\newcommand{\dcbracket}{{\}}\hspace*{-0.2em}{\}}}
\newcommand{\ttyperefine}[3]{\dobracket{#1}:~{#2}~|~{#3}\dcbracket}
\newcommand{\istrue}[1]{\tequal{{#1}}{\ttrue}}
\newcommand{\isfalse}[1]{\tequal{{#1}}{\tfalse}}
\newcommand{\tsum}[2]{{#1} + {#2}}
\newcommand{\tprod}[3]{\Sigma {#1}: {#2}.\ {#3}}
\newcommand{\tarrow}[3]{\Pi {#1}: {#2}.\ {#3}}
\newcommand{\texists}[3]{\exists {#1}: {#2}.\ {#3}}
\newcommand{\tforall}[3]{\forall {#1}: {#2}.\ {#3}}
\newcommand{\tintersect}[2]{{#1} \cap {#2}}
\newcommand{\tunion}[2]{{#1} \cup {#2}}
\newcommand{\tequal}[2]{{#1} \equiv {#2}}
\newcommand{\ttop}{\top}
\newcommand{\tbot}{\bot}
\newcommand{\indexedtype}[3]{{\tt Rec}({#1})({#2} \Rightarrow {#3})}
\newcommand{\rectype}[2]{{\tt Rec}({#1} \Rightarrow {#2})}
\newcommand{\tpoly}[2]{\forall {#1}: {\tt Type}.\,{#2}}

\newcommand{\tuletintype}[3]{{\tt Let}~{#1}~=~{#2}~{\tt in}~{#3}}
\newcommand{\tmatchintype}[3]{{\tt Match}({#1},{#2},{#3})}
\newcommand{\tsummatchintype}[3]{{\tt Either\_Match}({#1},{#2},{#3})}
\newcommand{\titeintype}[3]{{\tt If}\ {#1}\ {\tt Then}\ {#2}\ {\tt Else}\ {#3}}
\newcommand{\basetype}[2]{{\sf basetype}_{#1}({#2})}
\newcommand{\synval}[1]{{\tau_\alpha}}
\newcommand{\unify}{{\tt unify}}

\newcommand{\interp}{\theta}
\newcommand{\candidates}{{\sf Candidates}}
\newcommand{\candidate}{{\sf C}}

\newcommand{\tlist}[1]{List[{#1}]}
\newcommand{\tnil}[1]{Nil[{#1}]}
\newcommand{\tcons}[3]{Cons({#2},{#3})}

\newcommand{\fv}[1]{{\sf fv}({#1})}
\newcommand{\funs}[1]{{\sf funs}({#1})}
\newcommand{\sigdom}{{\sf dom}}
\newcommand{\sigcodom}{{\sf codom}}
\newcommand{\semantics}[1]{\llbracket~#1~\rrbracket}
\newcommand{\subst}[3]{{#1}[{#2}~\mapsto~{#3}]}
\newcommand{\erase}[1]{{\sf erase}({#1})}
\newcommand{\buildnat}{{\sf buildNat}}

\newcommand{\tspos}[2]{{\tt spos}({#1},{#2})}
\newcommand{\tpos}[2]{{\tt pos}({#1},{#2})}
\newcommand{\tneg}[2]{{\tt neg}({#1},{#2})}
\newcommand{\genplus}[3]{{\tt gen}^{\subseteq}_{#1}({#3})({#2})}
\newcommand{\genminus}[3]{{\tt gen}^{\supseteq}_{#1}({#3})({#2})}

\newcommand{\hastype}[3]{{#1} \vdash {#2}: {#3}}
\newcommand{\denote}[3]{{#1} \models {#2}: {#3}}
\newcommand{\redd}[3]{{#1} \models_{\sf Red} {#2}: {#3}}
\newcommand{\ljudge}[4]{{#1} \vdash_{#4} {#2}: {#3}}
\newcommand{\subtype}[3]{{#1} \vdash {#2} <: {#3}}
\newcommand{\areequal}[3]{{#1} \vdash {#2} \equiv {#3}}
\newcommand{\istype}[2]{{#1} \vdash {#2}\ \textit{type}}
\newcommand{\iscontext}[1]{\vdash {#1}\ \textit{context}}
\newcommand{\infertype}[3]{{#1} \vdash {#2}\Uparrow {#3}}
\newcommand{\checktype}[3]{{#1} \vdash {#2}\Downarrow {#3}}
\newcommand{\tinfer}[3]{{#1} \vdash {#2} \Uparrow {#3}}

\newcommand{\ltype}{{\sf ftype}}
\newcommand{\lpre}{{\sf pre}}
\newcommand{\lpost}{{\sf post}}
\newcommand{\lrank}{{\sf rank}}
\newcommand{\lbody}{{\sf body}}
\newcommand{\lt}{\prec}

\newcommand{\hole}{\mathcal{H}}
\newcommand{\isval}[1]{{#1} \in \Val}
\newcommand{\contextsymb}{\mathcal{E}}
\newcommand{\context}[1]{\contextsymb[{#1}]}
\newcommand{\thole}{\mathcal{H}}
\newcommand{\smallstep}{\hookrightarrow}
\newcommand{\err}{{\sf err}}
\newcommand{\equivalent}[2]{{#1} \approx {#2}}

\newcommand{\prop}{P}
\newcommand{\preserve}[1]{\overrightarrow{#1}}
\newcommand{\openpreserve}[1]{\overleftrightarrow{#1}}

\newcommand{\nodiv}[2]{\mathcal{N}({#1},{#2})}
\newcommand{\nodivterm}[2]{\mathcal{M}({#1},{#2})}
\newcommand{\mut}{\equiv}

\newcommand{\seprules}{\vspace{2ex}}

\newcommand{\psidebyside}[3]{
    \begin{minipage}{0.49\textwidth}
    #2
    \end{minipage} \hspace{#1}
    \begin{minipage}{0.49\textwidth}
    #3
    \end{minipage}
}
\newcommand{\sidebyside}[2]{\psidebyside{4em}{#1}{#2}}

\newcommand{\psidebysidebyside}[4]{
    \begin{minipage}{0.33\textwidth}
    #2
    \end{minipage} \hspace{#1}
    \begin{minipage}{0.32\textwidth}
    #3
    \end{minipage} \hspace{#1}
    \begin{minipage}{0.33\textwidth}
    #4
    \end{minipage}
}
\newcommand{\sidebysidebyside}[3]{\psidebyside{2em}{#1}{#2}{#3}}

\newcommand{\sidebysidebysidebyside}[4]{
    \begin{minipage}{0.24\textwidth}
    #1
    \end{minipage} \hspace{0.1em}
    \begin{minipage}{0.24\textwidth}
    #2
    \end{minipage} \hspace{0.1em}
    \begin{minipage}{0.24\textwidth}
    #3
    \end{minipage} \hspace{0.1em}
    \begin{minipage}{0.24\textwidth}
    #4
    \end{minipage}
}

\newcommand{\redvalues}[2]{\llbracket #1 \rrbracket_{\textsf{v}}^{#2}}
\newcommand{\redexpr}[2]{\llbracket #1 \rrbracket_{\textsf{t}}^{#2}}

\newcommand{\algo}{\mathcal{A}}

\newcommand{\lembed}{\lfloor}
\newcommand{\rembed}{\rfloor}
\newcommand{\embed}[1]{\lembed{#1}\rembed}
\newcommand{\cembed}[2]{\lembed{#1}\rembed_{#2}}

\newcommand{\lub}[2]{{#1} \sqcup {#2}}
\newcommand{\subalgo}[3]{{#1} \vdash {#2} \sqsubseteq {#3}}

\newcommand{\Stream}[1]{{\tt Stream}[{#1}]}
\newcommand{\IStream}[2]{{\tt Stream}_{#2}[{#1}]}
\newcommand{\List}{{\tt List}}
\newcommand{\IList}[1]{{\tt List}_{#1}}
\newcommand{\shead}[1]{{#1}.{\tt head}}
\newcommand{\stail}[1]{{#1}.{\tt tail}()}

\newcommand{\URA}{{()}{\Rightarrow}}

\newcommand{\PList}[1]{{\tt List}[{#1}]}
\newcommand{\IPList}[2]{{\tt List}_{#2}[{#1}]}




\section{Introduction}

Automatically verifying the correctness of higher-order
programs is a long-standing problem that arises in most
programming languages and proof assistants.  Despite
extensive research in program
verifiers and proof assistants~\citep{IsabelleHOL,Coq,GordonMelham93HOL,DBLP:books/daglib/0022394,Harrison17HOLLightTutorial,DBLP:conf/itp/000110,norell2007towards,DBLP:conf/plpv/Brady13,LiquidHaskell,Swamy2013,dafny} 
there remain significant challenges and trade-offs in
checking safety and termination.
A motivation for our work are implementations
that verify polymorphic functional programs using SMT solvers~\cite{SuterETAL11SatisfiabilityModuloRecursivePrograms,LiquidHaskell}.
To focus on foundations, we look at simpler verifiers that do not perform invariant
inference and are mostly based on unfolding recursive
definitions and encoding of higher-order functions into SMT theories
\cite{SuterETAL11SatisfiabilityModuloRecursivePrograms,
VoirolKK15,DBLP:conf/pldi/BlancK15}.
A recent implementation of such a verifier
is the Stainless system~\cite{Stainless}, which handles a subset of 
Scala~\cite{OderskyETAL07ProgrammingScala}.
The goal of Stainless is to verify that function contracts hold and that all functions terminate.
It was shown~\cite{HupelKuncak16TranslatingScalaProgramsIsabelleHOLSystemDescription} how to map certain patterns of specified Scala programs into Isabelle/HOL. Whereas this approach ensures soundness, it does not reuse the reasoning of Stainless and it can verify only some of the programs that 
Stainless verifier can prove. The present paper
seeks to provide direct foundations for verification and termination checking of functional
programs with a rich set of features for purely functional programming including non-monotonic data types.
On the other hand, our calculus does not aspire to directly support effects 
for which there exist excellent other systems \cite{Swamy2013}.

The subtleties of ensuring function termination have been an initial impetus for the calculus we present.
Termination is desirable for many executable functions in
programs and is even more important in formal specifications.  
A non-terminating function definition such as $f(x) = 1 + f(x)$ could be
easily mapped to a contradiction and violate the conservative
extension principle for definitions. Yet termination in the presence of higher-order
functions and data types is challenging to ensure. For example,
when using non-monotonic recursive types, terms can diverge
even without the explicit use of recursive functions, as
illustrated by the following snippet of
Scala code:
\begin{lstlisting}
case class D(f: D $\Rightarrow$ Unit) // non-monotonic recursive type
def g(d: D): Unit = d.f(d)            // non-recursive function definition
g(D(g))                               // diverging term, reduces to D(g).f(D(g)) and then again to g(D(g))
\end{lstlisting}
Furthermore, even though the concept of termination for all function inputs is an intuitively clear
property, its modular definition is subtle:
a higher order function $g$ taking another function $f$ as an argument
should terminate when given any terminating function $f$, which, in
turn, can be applied to expressions involving further calls to $g$.
The quest for solid foundations for termination led us to
type theoretic 
techniques, where \emph{reducibility} method
has long been used
to show strong normalization of expressive calculi~\cite{TaitStrongNorm}, \cite[Chapter 6]{ProofsAndTypesBook}, \cite{harper2016practical}.
As a natural framework
for analyzing support for first-class functions with preconditions and
post-conditions we embraced the ideas of 
refinement dependent types similar to those in Liquid Haskell \cite{LiquidHaskell} with
refinement-based notion of subtyping.
To explain proof obligation generation in the higher-order case (including the question of
which assumptions should be visible when checking a given assertion),
we resorted to well-known dependent ($\Pi$) function types.
To support parametric polymorphism we
incorporated type quantifiers, as in System F \cite{girard1971extension,ProofsAndTypesBook}.
We found that the presence of refinement types allowed us to explain
soundness of well-founded recursion based on user-defined measures.
The recursion
in programs is thus not syntactically restricted as in, e.g., System F.
To provide expressive support for iterative unfolding of recursive functions,
we introduced rules to make function bodies available while type checking of recursive functions.
For recursive type definitions, many existing
systems introduce separate notions of inductive and
co-inductive definitions. We found this distinction 
less natural for developers and chose to support expressive recursive types
(without a necessary restriction to positive recursion)
using sized types \cite{DBLP:conf/itp/000110}. 
We draw inspiration from a number of existing systems, yet our solution
is a new sound combination of features that work nicely together.

We combined these features into a new type system, System FR,
which we present as a bidirectional type checking algorithm. 
The algorithm generates type checking
and type inference goals by traversing terms and types, until it reaches a point
where it has to check that a given term evaluates to {\tt true}. This typically
arises when we want to check that a term $t$ has a refinement type
$\trefine{x}{T}{b}$, which is the case when $t$ has type $T$, and when the term
$b$ evaluates to {\tt true} in the context where $x$ equals $t$. Following
the tradition of SMT-based verifiers \cite{DetlefsETAL98ExtendedStaticChecking,BarnettETAL04SpecProgrammingSystemOverview},
we use the term \emph{verification condition} (VC) to refer
to a term that should evaluate to {\tt true}.

We prove the soundness of our type system using a reducibility interpretation of types.
The goal of our verification system is to ensure that a given term belongs to the
semantic denotation of a given type. For simple types such as natural numbers, this denotation
is the set of untyped lambda calculus terms that 
evaluate, in a finite number of steps,  
to a non-negative integer.
For function types the denotation are, as is typical in reducibility approaches,
terms that, when applied to terms in denotation of argument type, evaluate to terms 
in the denotation of the result type.
Such denotation gives us a unified framework for function contracts expressed as refinement types. The approach ensures
termination of programs because the semantics of types
only contain terms that are terminating in call-by-value semantics.

We have formally proven using the Coq proof assistant~\cite{Coq} the soundness
of our typing algorithm, implying that when verification conditions generated for
checking that a term $t$ belongs to a type $T$ are semantically valid, the term
$t$ belongs to the semantic denotation of the type $T$.
The bidirectional typing algorithm handles the expressive types in a
deterministic and predictable way, which enables good and localized error
reporting to the user. 
To solve generated verification conditions, we use existing
implementation invoking the Inox solver\footnote{\url{https://github.com/epfl-lara/inox}} that reduces higher-order queries to the first-order
language of SMT solvers~\cite{VoirolKK15}.
Our semantics of types provides a definition of soundness for such solvers;
any solver that respects the semantics can be used with our verification condition generator.
Our bidirectional type checking algorithm thus becomes a new, trustworthy 
verification condition generator for Stainless. We were successful in verifying many existing
Stainless benchmarks using the new approach.


We summarize our contributions as follows:
\begin{itemize}
\item
    We present a rich type system, called System FR, that combines System F with dependent types,
    refinements, equality types, and recursive
    types~(Sections~\ref{sec:operational} and \ref{sec:definitions}).
\item
    We define a bidirectional type-checking algorithm for System FR
    (Section~\ref{sec:bidi}). Our algorithm generates
    verification conditions that are solved by the (existing) SMT-based
    solver Inox.
\item
    We prove\footnote{\url{https://github.com/epfl-lara/SystemFR/tree/oopsla2019}} soundness of our bidirectional
    type-checking algorithm that reduces program correctness to proving
    that certain formulas always evaluate to true (Section~\ref{sec:coq}).
    Our 
    formalization also supports additional expressive notions, such as infinite intersections and unions as well as
    refinement conditions given by non-emptiness of an arbitrary type.
\item
    We built a verification condition generator based on these foundations\iftoggle{arxiv}{\footnote{\url{https://github.com/jad-hamza/stainless/tree/type-inference}}}{}
    and evaluated it on around $14$k lines of benchmarks
    (Section~\ref{sec:implementation}), showing that generating proof obligations
    using type checking is effective in practice.
\end{itemize}






\begin{figure}[tbhp]
\begin{center}
\begin{tabular}{c}
\begin{lstlisting}
def f(x: $\tau_1$): $\tau_2$ = {
  require($pre$[x])
  decreases($m$[x])
  $E$[x, f]
} ensuring { res $\RA$ $post$[x, res] }
\end{lstlisting}
\end{tabular}
\end{center}

\caption{Template of a recursive function with user-given contracts and a decreasing measure.}
\label{fig:template}
\end{figure}

\section{Examples of Program Verification and Termination Checking}
\label{sec:motiv}

Our goal is to verify correctness and termination of pure Scala
functions written as in Figure~\ref{fig:template}.
{\tt $pre$[x]} is the precondition of the function {\tt f}, and is written
by the user in the same language as the body of {\tt f}. The precondition may
contain arbitrary expressions and calls to other functions. Similarly, the user
specifies in $post$ the property that the results of the function should satisfy. To
ensure termination of {\tt f} (which might call itself recursively), the user
may also provide a measure using the {\tt \bf decreases} keyword, which is also an
expression (of type $\tnat$, the type of natural numbers) written in the same
language.
$\tau_1$ and $\tau_2$ may be arbitrary types, including
function types or algebraic data types. Informally, the function is terminating
and correct, if, for every value {\tt v} of type $\tau_1$ such that {\tt
$pre$[v]} evaluates to $\ttrue$, {\tt f(v)} returns (in a finite number of
steps) a value {\tt res} of type $\tau_2$ such that {\tt $post$[v,res]}
evaluates to $\ttrue$.
By using dependent and refinement types, this can be summarized by saying that
the function {\tt f} has type:
$
  \tarrow{x}{\trefine{x}{\tau_1}{pre[x]}}{\trefine{res}{\tau_2}{post[x,res]}}.
$


\begin{figure}[tbhp]
\begin{tabular}{c@{\hspace*{8mm}}c}

\begin{minipage}[b]{0.45\linewidth}
\begin{lstlisting}
sealed abstract class List
case object Nil extends List
case class Cons(head: $\tbigint$, tail: List) extends List

def filter(l: List, p: $\tbigint$ $\RA$ Boolean): List = {
  decreases($\tsize{\tt l}$)
  l match {
    case Nil $\RA$ Nil
    case Cons(h, t) if p(h) $\RA$ Cons(h, filter(t, p))
    case Cons(_, t) $\RA$ filter(t, p)  }  }

def count(l: List, x: $\tbigint$): $\tbigint$ = {
  decreases($\tsize{\tt l}$)
  l match {
    case Nil $\RA$ 0
    case Cons(h, t) $\RA$ (if (h == x) 1 else 0) 
                  + count(t, x) }}
\end{lstlisting}
\caption{The function {\tt filter} filters elements of a list based on a
predicate {\tt p}, and {\tt count} counts the number of occurrences of {\tt x}
in a list.
\label{fig:liblist}}
\end{minipage}
& 
\begin{minipage}[b]{0.43\linewidth}
\begin{lstlisting}
  def partition(
    l: $\PList{\tbigint}$,
    p: $X \RA \tbool$
  ): ($\PList{\tbigint}$, $\PList{\tbigint}$) = {
    decreases($\tsize{\tt l}$)
    l match {
      case Nil $\RA$ Nil
      case x :: xs $\RA$
        val (l1, l2) = partition(xs, p)
        if (p(x)) (x :: l1, l2)
        else (l1, x :: l2) }
  } ensuring { res $\RA$ 
      res._1 == filter(l, p) && 
      res._2 == filter(l, x $\RA$ !p(x))}
\end{lstlisting}  
\caption{A partition function specified using {\tt filter} and with 
termination measure is given with {\tt size}.}
\label{fig:partition}
\end{minipage}

\end{tabular}
\end{figure}

\begin{figure}

\begin{center}
\begin{tabular}{c}
\begin{lstlisting}
def partitionMultiplicity(@induct l: $\PList{\tbigint}$, p: $\tbigint \RA \tbool$, $x: \tbigint$): Boolean = {
  val (l1, l2) = partition(l, p)
  count(l, x) == count(l1, x) + count(l2, x)
} holds
\end{lstlisting}
\end{tabular}
\end{center}

\caption{A proof (by induction on {\tt l}) that partitioning a list preserves the
multiplicity of each element.
\label{fig:partition-multiplicity}}
\end{figure}


\begin{figure}[htb]

\begin{tabular}{c@{\hspace*{8mm}}c}

\begin{minipage}[h]{0.45\linewidth}
\begin{lstlisting}
def isSorted(l: List): Boolean = {
  decreases(size(l))
  l match {
    case Nil() $\RA$ true
    case Cons(x, Nil()) $\RA$ true
    case Cons(x, Cons(y, ys)) $\RA$
      x $\leq$ y && isSorted(Cons(y, ys))
  }
}
\end{lstlisting}
\end{minipage}
&
\begin{minipage}[h]{0.45\linewidth}
\begin{lstlisting}
def merge(l1: List, l2: List): List = {
  require(isSorted(l1) && isSorted(l2))
  decreases(size(l1) + size(l2))
  (l1, l2) match {
    case (Cons(x, xs), Cons(y, ys)) $\RA$
      if (x $\leq$ y) Cons(x, merge(xs, l2))
      else Cons(y, merge(l1, ys))
    case (Cons(_, _), Nil) $\RA$ l1
    case _ => l2 }
} ensuring { res => isSorted(res) }
\end{lstlisting}

\end{minipage}

\end{tabular}

\caption{A function that checks whether a list is sorted and a
  function that merges two sorted lists
  }
\label{fig:merge}

\end{figure}





As an example, consider the list type as defined in Figure~\ref{fig:liblist}.
We use $\tbigint$ to denote the type of integers (corresponding 
to Scala's BigInt in actual source code). The
function {\tt filter} filters elements from a list, while {\tt count} counts the
number of occurrences of an integer in the list. These two functions have no pre- or
postconditions. The {\tt decreases} clauses specify that the functions
terminate because the \emph{size} of the list decreases at each recursive call.

Using these functions we define {\tt partition} in Figure~\ref{fig:partition},
which takes a list {\tt l} of integers and partitions it according to a
predicate {\tt p: $\tbigint \Rightarrow \tbool$}. We prove in the postcondition
that partitioning coincides with applying {\tt filter} to the list with {\tt p}
and its negation.

Figure~\ref{fig:partition-multiplicity} shows a theorem
that {\tt partition} also preserves the multiplicity of each element. We
use here {\tt count} to state the property, but we could have used multisets
instead (a type which is natively supported in Stainless). The {\tt \bf
holds} keyword is a shorthand for {\tt {\bf ensuring} \{ res => res \}}. The
{\tt @induct} annotation instructs the system to add a recursive call to {\tt
partitionMultiplicity} on the tail of {\tt l} when {\tt l} is not empty. This
gives us access to the multiplicity property for the tail of {\tt l}, which the
system can then use automatically to prove that the property holds for {\tt l}
itself. This corresponds to a proof by induction on {\tt l}.

Figure~\ref{fig:merge} shows a function {\tt isSorted} that checks whether a
list is sorted, and a function {\tt merge} that combines two sorted lists in a
sorted list.
When given the above input, the system proves the termination
of all functions, establishes that postconditions of functions hold, and
shows that the theorem holds, without any user interaction or additional
annotations. For the {\tt merge} function, the postcondition
might seem too weak to establish that e.g.~{\tt Cons(x, merge(xs, l2))}
is sorted just based on the fact that {\tt merge(xs, l2)} is
sorted. However, since we put the definition of {\tt merge}
\emph{while} type-checking the body of {\tt merge} in the
context, it is possible to establish that {\tt x} is smaller
than the head of {\tt merge(xs, l2)}.
We give more details on this feature of our
system, called body-visible recursion, in
Section~\ref{sec:selfaware}.



\subsection{Reasoning about Streams}


\begin{figure}

\begin{lstlisting}
def constant[X](@erasable n: $\tnat$, x: X): @*\Stream{X}(n)*@ = {
  decreases(n)
  Stream(n)(x, $\URA$ constant[X](n-1,x))) }
\end{lstlisting}

\caption{Constant stream\label{fig:constant}}

\begin{lstlisting}
def zipWith[X,Y,Z](@erasable n: $\tnat$,
                   f: X $\RA$ Y $\RA$ Z, s1: @*\Stream{X}(n)*@, s2: @*\Stream{Y}(n)*@): @*\Stream{Z}(n)*@ = {
  decreases(n)
  Stream[Z](n)(f (@*\shead{s1}*@) (@*\shead{s2}*@), $\URA$ zipWith[X,Y,Z](n-1, f, @*\stail{s1},\ \stail{s2}*@) ) }
\end{lstlisting}
\caption{Zip function that combines elements of two streams using a two-argument function {\tt f}
\label{fig:zipwith}}

\begin{lstlisting}
def fib(@erasable n: $\tnat$): $\Stream{\tbigint}(n)$ = {
  decreases(n)
  Stream[$\tbigint$](n)(0, $\URA$ Stream[$\tbigint$](n-1)(1, $\URA$
    zipWith[$\tbigint,\tbigint,\tbigint$](n-2, plus, fib(n-2), @*\stail{fib(n-1)}*@))) }
\end{lstlisting}

\caption{Fibonacci stream defined using zipWith\label{fig:fibonacci}}
\end{figure}

Our system also supports reasoning about infinite data structures, including
streams that are computed on demand. These data structures are challenging to deal with because even defining termination
of an infinite stream is non-obvious, especially in absence of a concrete operation that uses
the stream. Given some type {\tt X}, $\Stream{\tt
X}$ represents the type of infinite streams containing elements in {\tt X}. In
a mainstream call-by-value language such as Scala, this type can be defined as:
\begin{lstlisting}
  case class @*\Stream{X}*@(head: X, tail: $\URA$ @*\Stream{X}*@)
\end{lstlisting}
For the sake of concise syntax, we typeset a function taking unit, {\tt
(u:Unit)=>e}, using Scala's syntax {\tt $\URA$e} for a function of zero
parameters. Given a stream {\tt s}: $\Stream{\tt X}$, we can call {\tt s.head}
to get the head of the stream (which is of type {\tt X}), or {\tt s.tail} to get
the tail of the stream (which is of type $\URA \Stream{\tt X}$). We can use
recursion to define streams, as shown in figures~\ref{fig:constant},
\ref{fig:fibonacci}, \ref{fig:zipwith}. The {\tt @erasable} annotation is used to
mark the erasable parameters {\tt n} of these functions. These parameters are used
as annotations to guide our type-checker, but they do not influence the
computation and can be erased at runtime. For instance, an erased version of
{\tt constant} (without erasable code and without type annotation) looks like:
\begin{lstlisting}
def constant(x) = Stream(x, $\URA$ constant(x))
\end{lstlisting}
Informally, we can say that the {\tt constant} stream is \emph{terminating}.
Indeed, it has the interesting property that, despite the recursion, for every
$n \in \Nat$, we can take the first $n$ elements in finite time (no divergence
in the computation). We say that {\tt constant(x)} is an
\emph{$n$-non-diverging stream}. Moreover, when a stream is $n$-non-diverging
for every $n \in \Nat$, we simply say that it is \emph{non-diverging}, which
means that we can take as many elements as we want without diverging, which is
the case for {\tt constant(x)}. Note that non-divergence of {\tt constant}
cannot be shown by defining a measure on its argument {\tt x} that strictly
decreases on each recursive call, because {\tt constant} is called recursively
on the exact same argument {\tt x}. Instead, we define a measure on the erasable
argument {\tt n} of the annotated version. This corresponds to using
type-based
termination~\citep{DBLP:journals/lmcs/Abel08,DBLP:phd/de/Abel2007,BartheGR08},
where the type of the function for the recursive call is smaller than the type
of the caller. We expand on that technique in Section~\ref{sec:measures}.

In the annotated version of {\tt constant} from Figure~\ref{fig:constant}, the
notation $\Stream{\tt X}(n)$ stands for streams of elements in {\tt X} which are
{\tt n}-non-diverging. The type of {\tt constant} then states that constant can
be called with any (erasable) parameter {\tt n} to build an {\tt n}-non-diverging
stream. Since parameter {\tt n} is computationally irrelevant, this proves that
the erased version of {\tt constant} returns a non-diverging stream.

We now give a variant of {\tt constant} which is \emph{diverging}:
\begin{lstlisting}
def badConstant(x) = Stream(x, $\URA$ badConstant(x).tail())
\end{lstlisting}

Indeed, given some {\tt x}, a call to {\tt badConstant(x).tail()}
ends up evaluating {\tt badConstant(x).tail()} again, and would
diverge.

The {\tt zipWith} function in Figure~\ref{fig:zipwith} takes two streams and a
function {\tt f}. It creates a new stream by applying {\tt f} to pairs of
elements taken from each stream. For {\tt zipWith}, we can verify that as long
as {\tt f} terminates on every input, and {\tt s1} and {\tt s2} are
non-diverging streams, then {\tt zipWith} returns a non-diverging stream.
We can then use {\tt zipWith} to define the well-known Fibonacci stream
(Figure~\ref{fig:fibonacci}), an infinite stream containing the
Fibonacci sequence: 0, 1, 1, 2, 3, 5, 8, 13, etc. We make use of a function {\tt
plus:$\tbigint \RA \tbigint \RA \tbigint$} that computes the sum of two
integers.
Just like the {\tt constant} streams, {\tt fib} is a non-diverging stream. For
instance, calling {\tt fib.tail().tail().tail().head}
returns (in finite time) the number $2$ from the Fibonacci sequence.

The important property that the type signature of {\tt zipWith} ensures is 
that, for every $n \in
\tnat$, if {\tt s1} and {\tt s2} are $n$-non-diverging streams, then {\tt
zipWith f s1 s2} is $n$-non-diverging as well. Our type system can 
check this property and then use it to make sure that
the definition of {\tt fib} type-checks.
We can also prove further properties of interest, e.g., that
 zipping two streams {\tt s1} and {\tt s2} with the
function {\tt (x:$\tnat$)$\RA$(y:$\tnat$)$\RA$x} 
returns a stream that behaves as the stream {\tt s1}.



\begin{figure}[htb]

\begin{flalign*}
t \Coloneqq\ & x\ |\ \tu \hspace{1em} | \hspace{1em}
              \tulambda{x}{t} \ |\ \tapp{t}{t}\hspace{1em} | \hspace{1em}
              (t,t) \ |\ \tprojl{t} \ |\ \tprojr{t}\ |\ \\
             & \tuleft{t} \ |\ \turight{t} \ |\ \tsummatch{t}{x \Rightarrow t}{x \Rightarrow t}\hspace{1em} | \hspace{1em}
             \ttrue\ |\ \tfalse\ |\ \tite{t}{t}{t}\ |\ \\
             & \tzero\ |\ \tsucc{t}\ |\
             \turec{t}{t}{(n,y) \Rightarrow t} \ |\
             \tufix{y \Rightarrow t} \ |\
             \tmatch{t}{t}{n \Rightarrow t} \\
             & \tufold{t} \ |\ \tunfoldin{t}{x \RA t}\hspace{1em} | \hspace{1em}
              \tuabs{t} \ |\ \tuinst{t}\hspace{1em} | \hspace{1em}
              \tuerr\ |\ \tulet{x}{t}{t}\ |\ \tsize{t}
\end{flalign*}

\caption{
    Grammar for untyped lambda calculus terms
}
\label{fig:uterms}
\end{figure}

\section{Syntax and Operational Semantics}
\label{sec:operational}

We now give a formal syntax for terms and show call-by-value operational
semantics. This untyped lambda calculus with pairs, tagged unions, integers,
booleans, and error values models programs that our verification system
supports. It is Turing complete and rather conventional.

\subsection{Terms of an Untyped Calculus}


Let $\Var$ be a set of \emph{variables}. We let $\UTerms$ be the set of all
(untyped) terms (see Figure~\ref{fig:uterms}) which includes the unit term
$\tu$, pairs, booleans, natural numbers, a recursor {\tt rec} for iterating over
natural numbers, a pattern matching operator {\tt match} for natural numbers, a
recursion operator {\tt fix}, an error term $\tuerr$ to represent crashes. The
recursor {\tt rec} can be simulated using {\tt fix} and {\tt match} but we keep
it in the paper for presenting examples.

The terms $\tufold{t}$ and $\tunfoldin{t_1}{x \RA t_2}$ are used to represent data structures
(such as lists or streams), where `$\tufold{}{}$' plays the role of a constructor,
and `$\tunfoldin{}{}$' the role of a deconstructor. The terms $\tuabs{t}$ and
$\tuinst{t}$ are used to represent the erasure of type abstractions and type
instantiation terms (for polymorphism) of the form $\tabs{\alpha}{t}$ and
$\tinst{t}{\tau}$, where $\alpha$ is a type variable and $\tau$ is a type. These
annotated terms will be introduced in a further section.

The term $\tsize{t}$ is a special term to internalize the sizes of syntax trees
of values (ignoring lambdas) of our language. It is used for measure of
recursive functions such as the map examples on lists shown in
Section~\ref{sec:motiv}.

Given a term $t$ we denote $\fv{t}$ the set of all \emph{free variables} of $t$.
Terms are considered up to renaming of locally bound variables
(\emph{alpha-renaming}).

\subsection{Call-by-Value Operational Semantics}


\begin{figure}[ht]

\centering
\footnotesize

\sidebyside{
    \infer[(\beta_1)]{
        \tprojl{(v_1,v_2)} \smallstep v_1
    }{
        \isval{v_1} &
        \isval{v_2}
    }
}{
    \infer[(\beta_2)]{
        \tprojr{(v_1,v_2)} \smallstep v_2
    }{
        \isval{v_1} &
        \isval{v_2}
    }
}

\seprules

\sidebyside
{
    \infer[(\beta_3)]{
        \tapp{(\tulambda{x}{t})}{v} \smallstep \subst{t}{x}{v}
    }{
        \isval{v}
    }
}
{
    \infer[(\beta_4)]{
        \tuinst{(\tuabs{t})} \smallstep t
    }{
    }
}

\seprules

\sidebyside
{
    \infer[(\beta_5)]{
       \tite{\ttrue}{t_1}{\_} \smallstep t_1
    }{
    }
}
{
    \infer[(\beta_6)]{
       \tite{\tfalse}{\_}{t_2} \smallstep t_2
    }{
    }
}

\seprules

\
{
    \infer[(\beta_7)]{
        \turec{\tzero}{t_0}{\_} \smallstep t_0
    }{
    }
}

\vspace{1ex}

\
{
    \infer[(\beta_8)]{
       \turec{\tsucc{v}}{t_0}{
           (n,y) \Rightarrow t_s
        } \smallstep
            t_s[n \mapsto v, y \mapsto
            \tulambda{u}{\turec{v}{t_0}{
              (n,y) \Rightarrow t_s
            }}]
    }{
        \isval{v}
    }
}

\vspace{2ex}

\
{
    \infer[(\beta_9)]{
       \tufix{y \Rightarrow t}
            \smallstep
            t[y \mapsto
            \tulambda{u}{\tufix{y \Rightarrow t}}]
    }{
    }
}

\vspace{2ex}

\sidebyside
{
    \infer[(\beta_{10})]{
        \tumatch{\tzero}{t_0}{\_} \smallstep t_0
    }{
    }
}
{
    \infer[(\beta_{11})]{
       \tumatch{\tsucc{v}}{\_}{
           n \Rightarrow t_s
        } \smallstep
            t_s[n \mapsto v]
    }{
        \isval{v}
    }
}

\vspace{2ex}

\
{
    \infer[(\beta_{12})]{
        \tsummatch{\tuleft{v}}{x \Rightarrow t}{\_} \smallstep \subst{t}{x}{v}
    }{
        \isval{v}
    }
}

\vspace{2ex}

\
{
    \infer[(\beta_{13})]{
        \tsummatch{\turight{v}}{\_}{x \Rightarrow t} \smallstep \subst{t}{x}{v}
    }{
        \isval{v}
    }
}

\vspace{2ex}

\sidebyside
{
    \infer[(\beta_{14})]{
        \tulet{x}{v}{t} \smallstep \subst{t}{x}{v}
    }{
        \isval{v}
    }
}
{
    \infer[(\beta_{15})]{
        \tunfoldin{\tufold{v}}{x \Rightarrow t} \smallstep \subst{t}{x}{v}
    }{
        \isval{v}
    }
}

\seprules

\sidebyside
{
    \infer[(\beta_{16})]{
        \tsize{v} \smallstep \buildnat(\sizesemantics{v})
    }{
        \isval{v}
    }
}
{
    \infer[\textit{(Evaluation Context)}]{
        \context{t} \smallstep \context{t'}
    }{
        t \smallstep t'
    }
}

\caption{Small-step call-by-value operational semantics of untyped terms}
\label{fig:eval}
\end{figure}


\begin{figure}[htb]
\footnotesize
\[\begin{array}{rcl@{\hspace*{1cm}}rcl}
\sizesemantics{\tsucc{v}} &=& 1 + \sizesemantics{v} &
\sizesemantics{\tufold{v}} &=& 1 + \sizesemantics{v} \\
\sizesemantics{\turight{v}} &=& 1 + \sizesemantics{v} &
\sizesemantics{\tuleft{v}} &=& 1 + \sizesemantics{v} \\
\sizesemantics{(v_1,v_2)} &=& \sizesemantics{v_1} + \sizesemantics{v_2} &
\multicolumn{3}{c}{\sizesemantics{v} = 0, \mbox{\it in all other cases}} \\
\end{array}\]
\caption{Definition of $\sizesemantics{v}$ for a value $v$.}

\label{fig:sizesemantics}
\end{figure}

The set $\Val$ of \emph{values} of our language is defined (inductively) to be
$\tzero$, $\tu$, $\ttrue$, $\tfalse$, every variable $x$, every lambda
term $\tulambda{x}{t}$ or $\tuabs{t}$, the terms of the form $\tsucc{v}$ or
$\tufold{v}$ where $\isval{v}$, and the terms of the form $(v_1,v_2)$ where
$v_1,v_2 \in \Val$.

The call-by-value small-step relation between two terms $t_1,t_2 \in \UTerms$,
written $t_1 \smallstep t_2$, is standard for the most part and given in
Figure~\ref{fig:eval}. Given a term $t$ and a value $v$, $\subst{t}{x}{v}$
denotes the term $t$ where every free occurrence of $x$ has been replaced by
$v$.


To evaluate the fixpoint operator {\tt fix}, we use the rule $\tufix{y
\Rightarrow t} \smallstep t[y \mapsto \tulambda{()}{\tufix{y \Rightarrow t}}]$,
which substitutes the {\tt fix} under a lambda with unit argument. We do this
wrapping of {\tt fix} in a lambda term because we wanted all substitutions to be
values for our call-by-value semantics, and {\tt fix} is not. This also means
that, to make a recursive call within $t$, one has to use {\tt y()} instead of
{\tt y}.

To define the semantics of $\tsize{}$, we use a (mathematical) function
$\sizesemantics{}$ that returns the size of a value, ignoring lambdas for which
it returns $0$. The precise definition is given in
Figure~\ref{fig:sizesemantics}.

We make use in the operational semantics of an \emph{evaluation context}
$\contextsymb$, which specifies through a \emph{hole} $\thole$ the next place
where reduction can occur in a term. Inductively, an evaluation context
$\contextsymb$ must be of one of the following forms:
\begin{flalign*}
  & \thole \quad | \quad
    \contextsymb{}\ e \quad | \quad
    v\ \contextsymb{} \quad | \quad
    (\contextsymb{},e) \quad | \quad
    (v,\contextsymb{}) \quad | \quad
    \tprojl{\contextsymb} \quad | \quad
    \tprojr{\contextsymb} \quad | \quad
    \tsucc{\contextsymb} \quad | \quad \\
  & \turec{\contextsymb}{\_}{\_} \quad | \quad
    \tumatch{\contextsymb}{\_}{\_} \quad | \quad
    \tite{\contextsymb}{t}{t} \quad | \quad
    \tulet{x}{\contextsymb}{t} \quad | \quad
    \tsize{\contextsymb} \quad | \quad \\
  & \tufold{\contextsymb} \quad | \quad
    \tunfoldin{\contextsymb}{\_} \quad | \quad
    \tuinst{\contextsymb} \quad | \quad
    \tuleft{\contextsymb} \quad | \quad
    \turight{\contextsymb} \quad | \quad
    \tsummatch{\contextsymb}{\_}{\_}
\end{flalign*}

Given a term $t$, we denote by $\context{t}$ the context $\contextsymb$ where
the hole $\mathcal{H}$ has been replaced by $t$.

We denote by $\smallstep^*$ the reflexive and transitive closure of
$\smallstep$. A term $t$ is \emph{normalizing} if there exists a value $v$ such
that $t \smallstep^* v$.



\begin{figure}[htb]
\newcommand{\bitspace}{\hspace{1em}}
\begin{flalign*}
\tau \Coloneqq\ & \tunit\ |\ \tbool\ |\ \tnat\ |\ \ttop\hspace{1em} | \bitspace
      \tarrow{x}{\tau}{\tau}\ |\ \tprod{x}{\tau}{\tau} \ |\ \tsum{\tau}{\tau}\hspace{1em} | \bitspace
      \tforall{x}{\tau}{\tau} \ |\ \tpoly{\alpha}{\tau}\\
     & \indexedtype{n}{\alpha}{\tau}\hspace{1em} | \bitspace
     \trefine{x}{\tau}{t} \ |\ \tequal{t}{t}
\end{flalign*}

\caption{
    Grammar for types $\tau$, where $x \in \Var$ is a term
    variable, $\alpha \in \Var$ is a type variable ($t$ denotes annotated terms of Figure~\ref{fig:terms}
    that complete the mutually recursive definition)
}
\label{fig:types}
\end{figure}

\section{Types, Semantics and Reducibility}
\label{sec:definitions}











We give in Figure~\ref{fig:types} the grammar for the \emph{types} $\tau$ that
our verification system supports. Given two types $\tau_1$ and $\tau_2$, we use
the notation $\tau_1 \rightarrow \tau_2$ for $\tarrow{x}{\tau_1}{\tau_2}$ when
$x$ is not a free variable of $\tau_2$. Similarly, we use the notation $\tau_1
\times \tau_2$ for $\tprod{x}{\tau_1}{\tau_2}$ when $x$ is not a free variable of $\tau_2$.

For recursive types, we introduce the notation:
\[
    \rectype{\alpha}{\tau} \define\tforall{n}{\tnat}{\indexedtype{n}{\alpha}{\tau}}
\]
Then, the type of (non-diverging) streams informally introduced in Section~\ref{sec:motiv}
can be understood as a notation, when {\tt X} is a type, for:
$\Stream{\tt X} \define \rectype{\alpha}{{\tt X} \times (\tunit \rightarrow \alpha)}$.
Similarly, for a natural number $n$, the type of $n$-non-diverging streams
$\IStream{\tt X}{n}$ is a notation for $\indexedtype{n}{\alpha}{{\tt X} \times (\tunit \rightarrow \alpha)}$.
Using this notation, we can also define finite data structures such as lists of
elements from {\tt X}, as follows:
$\PList{\tt X} \define \rectype{\alpha}{\tsum{\tunit}{{\tt X} \times \alpha}}$.

We show in Section~\ref{sec:rectype} that these types indeed correspond to
streams and lists respectively.

Let $\Type$ be the set of all types. We define a (unary) logical relation on
types to describe terms that do not get stuck (e.g.~due to the error term
$\err$, or due to an ill-formed application such as `$\ttrue\ \tzero$') and that
terminate to a value of the given type. Our definition is inspired by the notion
of \emph{reducibility} or \emph{hereditary termination}~(see
e.g.~\cite{TaitStrongNorm,ProofsAndTypesBook,harper2016practical}), which we use
as a guiding principle for designing the type system and its extensions.







\subsection{Reduciblity for Closed Terms}

For each type $\tau$, we define in Figure~\ref{fig:reducibility} mutually
recursively the sets of \emph{reducible values} $\redvalues{\tau}{\interp}$ and
\emph{reducible terms} $\redexpr{\tau}{\interp}$. In that sense, a type $\tau$
can be understood as a \emph{specification} that some terms satisfy (and some do
not).

These definitions require an environment $\interp$, called an
\emph{interpretation}, to give meaning to type variables. Concretely, an
interpretation is a partial map from type variables to sets of terms.
An interpretation $\interp$ has the constraint that for every type variable
$\alpha \in \sigdom(\interp)$, $\interp(\alpha)$ is a \emph{reducibility
candidate} $\candidate$, which, in our setting, means that all terms in
$\interp(\alpha)$ are (erased) values. The set of all reducibility candidates is
denoted by $\candidates \subseteq 2^{\UTerms}$, and an interpretation $\interp$
is therefore a partial map in $\Var \pfun \candidates$.

When the interpretation has no influence on the definition, we may omit it. For
instance, for every $\interp \in (\Var \pfun \candidates)$, we have
$\redvalues{\tnat}{\interp} =
\set{\tzero,\tsucc{\tzero},\tsucc{\tsucc{\tzero}},\dots}$, so we can just denote
this set by $\redvalues{\tnat}{}$.

By construction, $\redvalues{\tau}{\interp}$ only contains (erased) \emph{values} (of
type $\tau$), while $\redexpr{\tau}{\interp}$ contains (erased) \emph{terms} that
reduce to a value in $\redvalues{\tau}{\interp}$. For example, a term in  $\redexpr{\tnat \rightarrow \tnat}{\interp}$ is
not only normalizing as a term of its own, but also normalizes whenever applied
to a value in $\redvalues{\tnat}{\interp}$.


\begin{figure}

\begin{flalign*}
& \redvalues{\alpha}{\interp} \define \interp(\alpha) \\
& \redvalues{\ttop}{\interp} \define \Val \\
& \redvalues{\tunit}{\interp} \define \set{\tu} \\
& \redvalues{\tbool}{\interp} \define \set{\ttrue,\tfalse} \\
& \redvalues{\tnat}{\interp} \define \set{\tzero,\tsucc{\tzero},\tsucc{\tsucc{\tzero}},\dots} \\
& \redvalues{\tarrow{x}{\tau_1}{\tau_2}}{\interp} \define
  \set{f \in \Val\ |\ \forall a \in \redvalues{\tau_1}{\interp}.\
    {f\ a} \in \redexpr{\subst{\tau_2}{x}{a}}{\interp} }  \\
& \redvalues{\tforall{x}{\tau_1}{\tau_2}}{\interp} \define
    \set{b \in  \Val \ |\
      \forall a \in \redvalues{\tau_1}{\interp}.\ b \in \redvalues{\subst{\tau_2}{x}{a}}{\interp}} \\
& \redvalues{\tprod{x}{\tau_1}{\tau_2}}{\interp} \define
   \set{(a,b)\ |\ a \in \redvalues{\tau_1}{\interp} \land b \in \redvalues{\subst{\tau_2}{x}{a}}{\interp}}\\
& \redvalues{\trefine{x}{\tau}{b}}{\interp} \define \set{a \in \redvalues{\tau}{\interp}\ |\ \subst{b}{x}{a} \smallstep^* \ttrue} \\
& \redvalues{\tsum{\tau_1}{\tau_2}}{\interp} \define
    \set{ \tuleft{v}\ |\ v \in \redvalues{\tau_1}{\interp} } \cup
    \set{ \turight{v}\ |\ v \in \redvalues{\tau_2}{\interp} }\\
& \redvalues{\tequal{t_1}{t_2}}{\interp} \define
  \set{\tu} {\textrm{ if
      $\equivalent{t_1}{t_2}$, and $\emptyset$ otherwise, where}} \\
  & \qquad\qquad\qquad (\equivalent{t_1}{t_2}) \define
   \big( \forall v \in \Val.\
           (t_1 \smallstep^* v \iff t_2 \smallstep^* v) \big)
  \\
& \redvalues{\tpoly{\alpha}{\tau}}{\interp} \define \{
    v \in \Val\ |\
        \forall \candidate \in \candidates.\
            \tuinst{v} \in \redexpr{\tau}{\subst{\interp}{\alpha}{\candidate}}
   \} \\
  & \redvalues{\indexedtype{t}{\alpha}{\tau}}{\interp} \define \big\{ \tufold{v}\  \big| \ 
    (t \smallstep^* \tzero \land v \in \redvalues{\basetype{\alpha}{\tau}}{\interp})\ \lor \\
& \qquad\qquad\qquad\qquad\qquad\qquad\quad
    \exists n \in \redvalues{\tnat}{}.\
        t \smallstep^* \tsucc{n} \land
    v \in \redvalues{\tau}{
        \subst
            {\interp}
            {\alpha}
            {\redvalues{\indexedtype{n}{\alpha}{\tau}}{\interp}}}
    \ \big\} \\
& \redexpr{\tau}{\interp} \define \set{t\ |\ \exists v \in \redvalues{\tau}{\interp}.\ t \smallstep^* v}
\end{flalign*}

\begin{flalign*}
  & \basetype{\alpha}{\tprod{x}{\tau_1}{\tau_2}} \define
    \tprod{x}{\basetype{\alpha}{\tau_1}}{\basetype{\alpha}{\tau_2}} \\
  & \basetype{\alpha}{\tsum{\tau_1}{\tau_2}} \define
    \tsum{\basetype{\alpha}{\tau_1}}{\basetype{\alpha}{\tau_2}} \\
  & \basetype{\alpha}{\tau} \define 
    \textrm{ if $\alpha \in \fv{\tau}$} \textrm{ then } \ttop
  \textrm{ else } \tau 
\end{flalign*}

\caption{Definition of reducibility for values and for terms for each type.
The function $\basetype{}{}$ is an auxiliary function, used in the base case of
the definition for recursive types.}
\label{fig:reducibility}

\end{figure}




The type $\trefine{x}{\tau}{b}$ represents the values $v$ of type $\tau$ for
which $\subst{b}{x}{v}$ evaluates to $\ttrue$. We use this type as a
building block for writing specifications (pre and postconditions).

The type $\tforall{x}{\tau_1}{\tau_2}$ represents the values that are in the
intersection of the types $\subst{\tau_2}{x}{a}$ when $a$ ranges over values of
type $\tau_1$. This type differs from $\tarrow{x}{\tau_1}{\tau_2}$ in the
sense that a value in $\tforall{x}{\tau_1}{\tau_2}$ belongs to
every $\subst{\tau_2}{x}{a}$ for $a$ in $\tau_1$, while a value
in $\tarrow{x}{\tau_1}{\tau_2}$ is a function that, when
applied to some $a$ in $\tau_1$, produces a value in
$\subst{\tau_2}{x}{a}$. From a value $b$ in
$\tforall{x}{\tau_1}{\tau_2}$, we can
build a value in $\tarrow{x}{\tau_1}{\tau_2}$ (namely, $\tulambda{x}{b}$),
while the other way around is not always possible.

The sum type $\tsum{\tau_1}{\tau_2}$ represents values that are either of the form
$\tuleft{v}$ where $v$ is a reducible value of $\tau_1$, or of the form
$\turight{v}$ where $v$ is a reducible value of $\tau_2$.

The set of reducible values for the equality type
$\redvalues{\tequal{t_1}{t_2}}{\interp}$ makes use of a notion of equivalence on
terms which is based on operational semantics. More specifically, we say that
$t_1$ and $t_2$ are \emph{equivalent}, denoted $\equivalent{t_1}{t_2}$, if for
every value $v$, we have $t_1 \smallstep^* v \iff t_2 \smallstep^* v$. Note that
this equivalence relation is defined even if we do not know anything about the
types of terms $t_1$ and $t_2$, and it ensures that if one of the terms reduces
to a value, then so does the other.

The type $\tpoly{\alpha}{\tau}$ is the polymorphic type from System F. The set
$\redvalues{\tpoly{\alpha}{\tau}}{\interp}$ is defined by using the environment
$\interp$ to bind the type variable $\alpha$ to an arbitrary reducibility
candidate.

We use the \emph{recursive type} $\indexedtype{n}{\alpha}{\tau}$ as a building
block for representing data structures such as lists of streams. The definition
of reducibility for the recursive type makes use of an auxiliary function
$\basetype{}{}$ that can be seen as an (upper) approximation of the recursive
type. Note that $\basetype{\alpha}{\tau}$ (defined at the bottom of
Figure~\ref{fig:reducibility}) removes the type variable $\alpha$ from $\tau$.

Our reducibility definition respects typical lemmas that are needed to prove the
soundness of typing rules, such as the following substitution lemma (see
\cite{girard1971extension} for the lemma on System F), which we have formally proven
(see also Section~\ref{sec:coq} below).

\begin{lemma}
\label{lem:reducibility-subst}

Let $\tau_1$ and $\tau_2$ be two types, and let $\alpha$ be a type variable that
may appear in $\tau_1$ but not in $\tau_2$. Let $\interp$ be a type interpretation.
Then, we have:
\[
    \redvalues{\tau_1}{\subst{\interp}{\alpha}{\redvalues{\tau_2}{\interp}}} =
    \redvalues{\subst{\tau_1}{\alpha}{\tau_2}}{\interp}
\]

\end{lemma}

\subsubsection{Well-Foundedness of the Reduciblity Definition}

We can show that the definition given Figure~\ref{fig:reducibility} is
well-founded by defining a lexicographic measure $({\tt typeNodes}(\tau), {\tt
index}(\tau)) \in \mathbb{N} \times \UTerms$ on types $\tau$. The function ${\tt
typeNodes}(\tau)$ returns the size of the syntactic tree of type $\tau$,
ignoring the terms that appear inside. This size roughly corresponds to the number of
(top-level) type constructors in the tree of $\tau$.
For example, ${\tt typeNodes}(\trefine{x}{\tau}{p}) = 1 + {\tt typeNodes}(\tau)$ (ignoring $p$),
${\tt typeNodes}(\tsum{\tau_1}{\tau_2}) = 1 + {\tt
typeNodes}(\tau_1) + {\tt typeNodes}(\tau_2)$, and ${\tt
typeNodes}(\tequal{t_1}{t_2}) = 0$. Ignoring the size of terms inside types
ensures that given a type $\tau$, a term variable $x$, and a term $a$, we have:
${\tt typeNodes}(\subst{\tau}{x}{a}) = {\tt typeNodes}(\tau)$. As a result,
the measure ${\tt typeNodes}(\tau)$ in the definition of reducibility decreases
for indexed types such as $\tprod{x}{\tau_1}{\tau_2}$ or
$\tarrow{x}{\tau_1}{\tau_2}$.

The number ${\tt typeNodes}(\tau)$
decreases in every case of Figure~\ref{fig:reducibility}, \emph{except} for recursive
types $\indexedtype{t}{\alpha}{\tau}$ where the measure stays the same in the
recursive call to the denotation on $\indexedtype{n}{\alpha}{\tau}$ with $t
\smallstep^* \tsucc{n}$.
This is where we use the second component of the lexicographic measure, ${\tt
index}(\tau)$. We define ${\tt index}(\indexedtype{t}{\alpha}{\tau}) = t$ and
${\tt index}(\tau) = \tuerr$ for every other type. Then, given
$t_1,~t_2~\in~\UTerms$, we consider $t_1$ to be (strictly) smaller than $t_2$ if
there exist $v_1, v_2 \in \redvalues{\tnat}{}$, such that $t_1 \smallstep^*
v_1$, $t_2 \smallstep^* v_2$, and $v_1$ is strictly smaller than $v_2$ when seen
as a natural number. Therefore, the second component ${\tt
index}(\indexedtype{n}{\alpha}{\tau}) = n$ is strictly smaller than ${\tt
index}(\indexedtype{t}{\alpha}{\tau}) = t$ in the definition (because $t
\smallstep^* \tsucc{n}$), which ensures that the overall lexicographic measure
decreases.


\subsection{Reduciblity for Open Terms}

Having defined reducibility for closed terms,
we now define what it means for a term $t$ with free term and type variables to
be reducible for a type $\tau$. Informally, we want to ensure that for every
interpretation of the type variables, and for every substitution of values for
the term variables, the term $t$ reduces in a finite number of steps to a value
in type $\tau$. This is formalized by a (semantic) typing relation
$\redd{\Theta;\Gamma}{t}{\tau}$ which is defined as follows.

First, a \emph{context} $\Theta;\Gamma$ is made of a finite set $\Theta
\subseteq \Var$ of type variables and of a sequence $\Gamma$ of pairs in $\Var
\times \Type$. The \emph{domain} of $\Gamma$, denoted $\sigdom(\Gamma)$ is the
list of variables (in $\Var$) appearing in the left-hand-sides of the pairs. We
implicitly assume throughout the paper that all variables appearing in the
domains are distinct. This enables us to use $\Gamma$ as a partial map from
$\Var$ to $\Type$. We use a sequence to represent $\Gamma$ as the order of
variables is important, since a variable may have a (dependent) type which
refers to previous variables in the context.

Given a partial map $\gamma \in \Var \pfun \UTerms$, we write $\gamma(t)$ for
the term $t$ where every variable $x$ is replaced by $\gamma(x)$. We use the
same notation $\gamma(\tau)$ for applying a substitution to a type $\tau$.

Given a context $\Theta;\Gamma$, a \emph{reducible substitution for
$\Theta;\Gamma$} is a pair of partial maps $\interp \in \Var \pfun \candidates$ and
$\gamma \in \Var \pfun \UTerms$ where:
  $\sigdom(\interp) = \Theta$,
  $\sigdom(\gamma) = \sigdom(\Gamma)$, and
  $\forall x \in \sigdom(\Gamma).\
  \gamma(x) \in \redvalues{\gamma(\Gamma(x))}{\interp}$.

Note that the substitution $\gamma$ is also applied to the type $\Gamma(x)$,
since $\Gamma(x)$ may be a dependent type with free term variables.
The set of all
pairs of reducible substitutions for $\Theta;\Gamma$ is denoted
$\redvalues{\Theta;\Gamma}{}$.

Finally, given a context $\Theta;\Gamma$, a term $t$ and a type $\tau$, we say
that $\redd{\Theta;\Gamma}{t}{\tau}$ holds when for every pair of substitutions
$\interp,\gamma$ for the context $\Theta;\Gamma$, $\gamma(t)$ belongs the
reducible values at type $\gamma(\tau)$. Formally,
$\redd{\Theta;\Gamma}{t}{\tau}$ is defined to hold when:
\[
  \forall \interp,\gamma \in \redvalues{\Theta;\Gamma}{}.\
    \gamma(t) \in \redexpr{\gamma(\tau)}{\interp}
\]

Our bidirectional type checking and inference algorithm in Section~\ref{sec:automated-verifier} is a sound (even if incomplete)
procedure to check
$\redd{\Theta;\Gamma}{t}{\tau}$.


\subsection{Recursive Types}
\label{sec:rectype}

We explain in this section how to interpret the type
$\indexedtype{n}{\alpha}{\tau}$ (see reducibility definition in
Figure~\ref{fig:reducibility}) and how the $\Stream{\tt X}$ and $\PList{\tt X}$
types represent streams and lists.




\subsubsection{Infinite Streams}


For a natural number $n$, consider the type $S_n \define \IStream{\tnat}{n}
\define \indexedtype{n}{\alpha}{\tnat \times (\tunit \rightarrow \alpha)}$. Let
us first see what $S_n$ represents for small values of $n$. As a shortcut, we
use the notations $0$, $1$, $2$, $\dots$ for $\tzero$, $\tsucc{\tzero}$,
$\tsucc{\tsucc{\tzero}}$, $\dots$

The definition $\redvalues{S_0}{}$ refers to $\basetype{\alpha}{\tnat \times
(\tunit \rightarrow \alpha)}$, which is $\tnat \times \ttop$ by
definition. This means that $\redvalues{S_0}{}$ is the set of values of the
form $\tufold{(a,v)}$, where $a \in \redvalues{\tnat}{}$, and $v \in \Val$.

By unrolling the definition, we get that $\redvalues{S_1}{}$ is the set of
values of the form $\tufold{v}$ where $v$ is in $\redvalues{\tnat \times (\tunit
\rightarrow \alpha)}{\subst{}{\alpha}{\redvalues{S_0}{}}}$, which is the same
(by Lemma~\ref{lem:reducibility-subst}) as $\redvalues{\tnat \times (\tunit
\rightarrow S_0)}{}$. Therefore, $\redvalues{S_1}{}$ is the set of values of the
form $\tufold{a, f}$ where $a \in \redvalues{\tnat}{}$ and $f \in
\redvalues{\tunit \rightarrow S_0}{}$. This means that when it is applied to
$\tu$, $f$ terminates and returns a value in $\redvalues{S_0}{\interp}$.
Similarly, $\redvalues{S_2}{}$ is the set of values of the form
$\tufold{a, f}$ where $n \in \redvalues{\tnat}{}$ and $f \in \redvalues{\tunit
\rightarrow S_1}{}$.

To summarize, we can say that for every $n \in \redvalues{\tnat}{}$, $S_n$
represents values of the language that behave as streams of natural numbers, as
long as they are unfolded at most $n+1$ times. This matches the property we
mentioned in Section~\ref{sec:motiv}, as $S_n$ represents the streams that are
$n+1$-non-diverging.
%
We can show that as $n$ grows, $S_n$ gets more and more constraints:
    $\redvalues{S_0}{} \supseteq \redvalues{S_1}{} \supseteq \redvalues{S_2}{}
    \supseteq \dots$
In the limit, a value $v \in \redvalues{\tforall{n}{\tnat}{S_n}}{}$ (which is in
every $S_n$ for $n \in \redvalues{\tnat}{}$), represents a stream of natural
numbers, that, regardless of the number of times it is unfolded, does not
diverge, i.e.~a non-diverging stream. Equivalently, we have
$v \in \redvalues{\Stream{\tnat}}{}$.


\subsubsection{Finite Lists}

Types of the form $\rectype{\alpha}{\tau}$ can also be used to
represent finite data structures such as lists.
We let $\IPList{\tt X}{n}$ be a notation for
$\indexedtype{n}{\alpha}{\tsum{\tunit}{{\tt X} \times \alpha}}$, so that:
\[
    \PList{\tt X} \define \tforall{n}{\tnat}{\IPList{\tt X}{n}}.
\]
Here are some examples to show how lists are encoded:
\begin{itemize}
\item The empty list is $\tufold{\tuleft{}}$,
\item A list with one element $n$ is $\tufold{\turight{n, \tufold{\tuleft{}}}}$,
\item Given an element $n$ and a list $l$, we can construct the
    list $n :: l$ by writing: $\tufold{\turight{n, l}}$.
\end{itemize}

Let us now see why $\PList{\tt X}$ represents the type of all finite lists of
elements in ${\tt X}$. The first thing to note is that given
$n \in \redvalues{\tnat}{}$, $\IPList{\tt X}{n}$ does \emph{not} represent the
lists of size $n$. For instance, we know that $\redvalues{\IPList{\tt X}{0}}{}$
is the set of values of the form $\tufold{v}$ where $v \in
\redvalues{\basetype{\alpha}{\tsum{\tunit}{{\tt X} \times \alpha}}}{}$, i.e.~$v
\in \redvalues{\ttop}{} = \Val$. Therefore, $\IPList{\tt X}{0}$ contains lists of
all sizes (and also all values that do not represent lists, such as
$\tufold{\tzero}$ or $\tufold{\tulambda{x}{(\tu}}$).

Instead, $\IPList{\tt X}{n}$ can be understood as the values that, as long as
they are unfolded no more than $n$ times, behave as lists. As for
streams, we have: $\redvalues{\IPList{\tt X}{0}}{} \supseteq
\redvalues{\IPList{\tt X}{1}}{}
\supseteq \redvalues{\IPList{\tt X}{2}}{} \supseteq \dots$ where the monotonicity
follows because $\alpha$ only appears in positive positions in the
definitions of the recursive types for streams and lists.
In the limit, we can show that $\PList{\tt X}$ contains all finite lists, and
nothing more.
\begin{restatable}{lemma}{reclists}
\label{lem:reclists}
Let $v \in \Val$ be a value and ${\tt X}$ be some type. Then, $v \in
\redvalues{\PList{\tt X}}{}$ if and only if there exists $k \geq 0$ and
$a_1,\dots,a_k \in \redvalues{\tt X}{}$ such that $v = \tufold{\turight{a_1,
\dots\tufold{\turight{a_k, \tufold{\tuleft{}}}}\dots}}$.
\end{restatable}
It may seem surprising that the type of streams $\rectype{\alpha}{{\tt X} \times
(\tunit \rightarrow \alpha)}$ contains infinite streams while the type of lists
$\rectype{\alpha}{\tsum{\tunit}{{\tt X} \times \alpha}}$ only contains finite
lists. The reason is that, in a call-by-value language, a value representing an
infinite list would need to have an infinite syntax tree, with infinitely many
$\tufold{}$'s (which is not possible). On the other hand, we can represent
infinite streams by hiding recursion underneath a lambda term as shown in
Section~\ref{sec:motiv}.


\section{A Bidirectional Type-Checking Algorithm}
\label{sec:automated-verifier}
\label{sec:bidi}

In this section, we give procedures for inferring a type $\tau$ for
a term $t$ in a context $\Theta;\Gamma$, denoted
$\infertype{\Theta;\Gamma}{t}{\tau}$, as well as for checking that the type of a term $t$
is $\tau$, denoted $\checktype{\Theta;\Gamma}{t}{\tau}$. We introduce rules of
our procedures throughout this section; the full set of rules is given in
figures~\ref{fig:bidirectional-infer} and ~\ref{fig:bidirectional-check}.


\begin{figure}

\centering

\footnotesize

\renewcommand{\seprules}{\vspace{6pt}}

\psidebysidebyside{0.6em}
{
  \infer[\textit{(Infer Var)}]{
    \infertype{\Theta;\Gamma}{x}{\tau}
  }{
    \Gamma(x) = \tau
  }
}
{
  \infer[\textit{(Infer True)}]{
    \infertype{\Theta;\Gamma}{\ttrue}{\tbool}
  }{
  }
}
{
  \infer[\textit{(Infer False)}]{
    \infertype{\Theta;\Gamma}{\tfalse}{\tbool}
  }{
  }
}

\seprules

\psidebysidebyside{0.6em}
{
  \infer[\textit{(Infer Unit)}]{
    \infertype{\Theta;\Gamma}{\tu}{\tunit}
  }{
  }
}{
  \infer[\textit{(Infer Zero)}]{
    \infertype{\Theta;\Gamma}{\tzero}{\tnat}
  }{
  }
}{
  \infer[\textit{(Infer Succ)}]{
    \infertype{\Theta;\Gamma}{\tsucc{t}}{\tnat}
  }{\checktype{\Theta;\Gamma}{t}{\tnat}}
}

\seprules

\
\infer[\textit{(Infer If)}]{
  \infertype{\Theta;\Gamma}{\tite{t_1}{t_2}{t_3}}{\titeintype{t_1}{\tau_2}{\tau_3}}
}{
  \checktype{\Theta;\Gamma}{t_1}{\tbool} &
  \infertype{\Theta;\Gamma,p:\tequal{t_1}{\ttrue}}{t_2}{\tau_2} &
  \infertype{\Theta;\Gamma,p:\tequal{t_1}{\tfalse}}{t_3}{\tau_3}
}

\seprules

\psidebysidebyside{0.7em}
{
  \scriptsize
  \infer[\textit{(Infer Left)}]{
    \infertype{\Theta;\Gamma}{\taleft{\tsum{\tau_1}{\tau_2}}{t}}{\tsum{\tau_1}{\tau_2}}
  }{
    \checktype{\Theta;\Gamma}{t}{{\tau_1}}
  }
}
{
  \scriptsize
  \infer[\textit{(Infer Right)}]{
    \infertype{\Theta;\Gamma}{\taright{\tsum{\tau_1}{\tau_2}}{t}}{\tsum{\tau_1}{\tau_2}}
  }{
    \checktype{\Theta;\Gamma}{t}{{\tau_2}}
  }
}
{
  \scriptsize
  \infer[\textit{(Infer Size)}]{
    \infertype{\Theta;\Gamma}{\tsize{t}}{\tnat}
  }{
    \infertype{\Theta;\Gamma}{t}{\tau}
  }
}

\seprules

{
\
  {
    \infer[\textit{(Infer Match)}]{
      \infertype
        {\Theta;\Gamma}
        {\tmatch{t_n}{t_0}{n \Rightarrow t_s}}
        {\tmatchintype{t_n}{\tau_1}{n \Rightarrow \tau_2}}
    }{
      \checktype{\Theta;\Gamma}{t_n}{\tnat} &
      \infertype{\Theta;\Gamma,p:\tequal{t_n}{\tzero}}{t_0}{\tau_1} &
      \infertype{\Theta;\Gamma,n:\tnat,p:\tequal{t_n}{\tsucc{n}}}{t_s}{\tau_2}
    }
  }
}

\seprules

{
\
  {
    \infer[\textit{(Infer Either Match)}]{
      \infertype
        {\Theta;\Gamma}
        {\tsummatch{t}{x \Rightarrow t_1}{x \Rightarrow t_2}}
        {\tsummatchintype{t}{x \Rightarrow \tau_1'}{x \Rightarrow \tau_2'}}
    }{
      \infertype{\Theta;\Gamma}{t}{\tsum{\tau_1}{\tau_2}} &
      \infertype{\Theta;\Gamma,x:\tau_1,p:\tequal{t}{\tuleft{x}}}{t_1}{\tau_1'} &
      \infertype{\Theta;\Gamma,x:\tau_2,p:\tequal{t}{\turight{x}}}{t_2}{\tau_2'}
    }
  }
}

\vspace{1ex}

\
{
  \infer[\textit{(Infer Rec)}]{
    \infertype{\Theta;\Gamma}{\trec{n \Rightarrow \tau}{t_n}{t_0}{(n,y) \Rightarrow t_s}}{\tuletintype{n}{t_n}{\tau}}
  }{
    \begin{array}{@{}c@{}}
    \checktype{\Theta;\Gamma}{t_n}{\tnat} \qquad
    \checktype{\Theta;\Gamma}{t_0}{\tau[n\mapsto\tzero]} \\
    \checktype{\Theta;\Gamma,n:\tnat,y:\tunit\rightarrow\tau,p:\tequal{y}{\tlambda{u}{\tunit}{\trec{n\rightarrow\tau}{n}{t_0}{(n,y) \Rightarrow t_s}}}}{t_s}{\tau[n\mapsto\tsucc{n}]}
    \end{array}
  }
}



\seprules

\
{
  \infer[\textit{(Infer Fix)}]{
    \infertype
      {\Theta;\Gamma}
      {\tfix{n \Rightarrow \tau}{(n,y) \Rightarrow t}}
      {\tforall{n}{\tnat}{\tau}}
  }{
    \begin{array}{@{}c@{}}
    n \notin \fv{\erase{t}}\\
        \hspace{-5em}\Theta;\Gamma,
        n:\tnat,
        y:\tunit\rightarrow \tforall{m}{\trefine{m}{\tnat}{m < n}}{\subst{\tau}{n}{m}}, \\
    \hspace{7em}    p:\tequal{y}{\tlambda{u}{\tunit}
          {\tfix{n \Rightarrow \tau}{(n,y) \Rightarrow t}}}
              \hspace{2em}
      \vdash
      {t} \Downarrow
      {\tau}
    \end{array}
  }
}

\seprules

\psidebyside{1em}{
  \infer[\textit{(Infer Pair)}]{
    \infertype{\Theta;\Gamma}{(t_1, t_2)}{\tprod{x}{\tau_1}{\tau_2}}
  }{
    \infertype{\Theta;\Gamma}{t_1}{\tau_1} &
    \infertype{\Theta;\Gamma}{t_2}{\tau_2}
  }
}
{
  \infer[\textit{(Infer Let)}]{
    \infertype
      {\Theta;\Gamma}
      {\tulet{x}{t_1}{t_2}}
      {\tuletintype{x}{t_1}{\tau_2}}
  }{
    \infertype{\Theta;\Gamma}{t_1}{\tau_1} &
    \infertype{\Theta;\Gamma,x:\tau_1,p:\tequal{x}{t_1}}{t_2}{\tau_2}
  }
}

\seprules

\sidebyside{
  \infer[\textit{(Infer Proj1)}]{
    \infertype{\Theta;\Gamma}{\tprojl{t}}{\tau_1}
  }{
    \infertype{\Theta;\Gamma}{t}{\tprod{x}{\tau_1}{\tau_2}}
  }
}{
  \infer[\textit{(Infer Proj2)}]{
    \infertype{\Theta;\Gamma}{\tprojr{t}}{\tuletintype{x}{\tprojl{t}}{\tau_2}}
  }{
    \infertype{\Theta;\Gamma}{t}{\tprod{x}{\tau_1}{\tau_2}}
  }
}

\seprules

\sidebyside
{
  \infer[\textit{(Infer Lambda)}]{
    \infertype{\Theta;\Gamma}{\tlambda{x}{\tau_1}{t}}{\tarrow{x}{\tau_1}{\tau_2}}
  }{
    \infertype{\Theta;\Gamma,x:\tau_1}{t}{\tau_2}
  }
}
{
  \infer[\textit{(Infer App)}]{
    \infertype{\Theta;\Gamma}{\tapp{t_1}{t_2}}{\tuletintype{x}{t_2}{\tau}}
  }{
    \infertype{\Theta;\Gamma}{t_1}{\tarrow{x}{\tau_2}{\tau}} &
    \checktype{\Theta;\Gamma}{t_2}{\tau_2}
  }
}

\seprules

\sidebyside
{
  \infer[\textit{(Infer Type Abs)}]{
    \infertype{\Theta;\Gamma}{\tabs{\alpha}{t}}{\tpoly{\alpha}{\tau}}
  }{
    \infertype{\Theta,\alpha;\Gamma}{t}{\tau}
  }
}
{
  \infer[\textit{(Infer Type App)}]{
    \infertype{\Theta;\Gamma}{\tinst{t}{\tau_2}}{\subst{\tau_1}{\alpha}{\tau_2}}
  }{
    \infertype{\Theta;\Gamma}{t}{\tpoly{\alpha}{\tau_1}}
  }
}

\seprules

\
{
  \infer[\textit{(Infer Forall Instantiation)}]{
    \infertype{\Theta;\Gamma}{\tinstforall{t_1}{t_2}}{\tuletintype{x}{t_2}{\tau}}
  }{
    \infertype{\Theta;\Gamma}{t_1}{\tforall{x}{\tau_2}{\tau}} &
    \checktype{\Theta;\Gamma}{t_2}{\tau_2}
  }
}

\seprules

\
{
  \infer[\textit{(Infer Fold)}]{
    \infertype
      {\Theta;\Gamma}
      {\tfold{\indexedtype{n}{\alpha}{\tau}}{t}}
      {\indexedtype{n}{\alpha}{\tau}}
  }{
    \begin{array}{@{}c@{}}
    \checktype{\Theta;\Gamma}{n}{\tnat} \hspace{2em}
    \checktype
      {\Theta;\Gamma, p: \tequal{n}{\tzero}}
      {t}
      {\basetype{\alpha}{\tau}} \\
    \checktype
      {\Theta;\Gamma; n': \tnat, p: \tequal{n}{\tsucc{n'}}}
      {t}
      {\subst{\tau}{\alpha}{\indexedtype{n'}{\alpha}{\tau}}}
    \end{array}
  }
}

\seprules

\
{
  \infer[\textit{(Infer Unfold)}]{
    \infertype
      {\Theta;\Gamma}
      {\tunfoldin{t_1}{x \RA t_2}}
      {\tau'}
  }{
    \begin{array}{@{}c@{}}
    \infertype{\Theta;\Gamma}{t_1}{\indexedtype{n}{\alpha}{\tau}}
      \hspace{4em}
    \infertype
      {\Theta;\Gamma,
        x: \basetype{\alpha}{\tau},
        p_1: \tequal{t_1}{\tufold{x}},
        p_2: \tequal{n}{\tzero}
      }
      {t_2}{\tau'}\\
    \infertype
      {\Theta;\Gamma,
        x: \subst{\tau}{\alpha}{\indexedtype{\tpred{n}}{\alpha}{\tau}},
        p: \tequal{t_1}{\tufold{x}}
      }
      {t_2}{\tau'}
    \end{array}
  }
}

\seprules

\
{
  \infer[\textit{(Infer Unfold Positive)}]{
    \infertype
      {\Theta;\Gamma}
      {\tunfoldin{t_1}{x \RA t_2}}
      {\tau'}
  }{
    \begin{array}{@{}c@{}}
    \infertype{\Theta;\Gamma}{t_1}{\indexedtype{n}{\alpha}{\tau}}
      \hspace{4em}
    \color{blue}\fbox{$\areequal{\Theta;\Gamma}{{\tt lessThan}\ 0\ n}{\ttrue}$} \\
    \infertype
      {\Theta;\Gamma,
        x: \subst{\tau}{\alpha}{\indexedtype{\tpred{n}}{\alpha}{\tau}},
        p: \tequal{t_1}{\tufold{x}}
      }
      {t_2}{\tau'}
    \end{array}
  }
}



\seprules

\sidebyside
{
  \infer[\textit{(Infer Err)}]{
    \infertype{\Theta;\Gamma}{\terr{\tau}}{\tau}
  }{
    \color{blue}\fbox{$\areequal{\Theta;\Gamma}{\ttrue}{\tfalse}$}
  }
}
{
  \infer[\textit{(Infer Refl)}]{
    \infertype{\Theta;\Gamma}{\trefl{t_1}{t_2}}{\tequal{t_1}{t_2}}
  }{
    \color{blue}\fbox{$\areequal{\Theta;\Gamma}{t_1}{t_2}$}
  }
}

\seprules

\
{
  \infer[\textit{(Infer Drop Refinement)}]{
    \infertype{\Theta;\Gamma}{t}{\tau}
  }{
    \infertype{\Theta;\Gamma}{t}{\trefine{x}{\tau}{p}}
  }
}

\caption{
  $\infertype{\Theta;\Gamma}{t}{\tau}$
  infers a type $\tau$ for $t$ in context $\Theta;\Gamma$ based on the shape of $t$.
  The \textit{(Infer Drop Refinement)} rule is applied
  with low priority, only if no other rule is applicable, keeping type checking
  deterministic.
}
\label{fig:bidirectional-infer}
\end{figure}


\begin{figure}

\centering

\footnotesize

\renewcommand{\seprules}{\vspace{6pt}}

\seprules

\
\infer[\textit{(Check If)}]{
  \checktype{\Theta;\Gamma}{\tite{t_1}{t_2}{t_3}}{\tau}
}{
  \checktype{\Theta;\Gamma}{t_1}{\tbool} &
  \checktype{\Theta;\Gamma,p:\tequal{t_1}{\ttrue}}{t_2}{\tau} &
  \checktype{\Theta;\Gamma,p:\tequal{t_1}{\tfalse}}{t_3}{\tau}
}

\seprules

{
\
  {
    \infer[\textit{(Check Match)}]{
      \checktype
        {\Theta;\Gamma}
        {\tmatch{t_n}{t_0}{n \Rightarrow t_s}}
        {\tau}
    }{
      \checktype{\Theta;\Gamma}{t_n}{\tnat} &
      \checktype{\Theta;\Gamma,p:\tequal{t_n}{\tzero}}{t_0}{\tau} &
      \checktype{\Theta;\Gamma,n:\tnat,p:\tequal{t_n}{\tsucc{n}}}{t_s}{\tau}
    }
  }
}

\seprules

{
\
  {
    \infer[\textit{(Check Either Match)}]{
      \checktype
        {\Theta;\Gamma}
        {\tsummatch{t}{x \Rightarrow t_1}{x \Rightarrow t_2}}
        {\tau}
    }{
      \infertype{\Theta;\Gamma}{t}{\tsum{\tau_1}{\tau_2}} &
      \checktype{\Theta;\Gamma,x:\tau_1,p:\tequal{t}{\tuleft{x}}}{t_1}{\tau} &
      \checktype{\Theta;\Gamma,x:\tau_2,p:\tequal{t}{\turight{x}}}{t_2}{\tau}
    }
  }
}

\seprules

\sidebyside
{
  \infer[\textit{(Check Let)}]{
    \checktype
      {\Theta;\Gamma}
      {\tulet{x}{t_1}{t_2}}
      {\tau}
  }{
    \infertype{\Theta;\Gamma}{t_1}{\tau_1} &
    \checktype{\Theta;\Gamma,x:\tau_1,p:\tequal{x}{t_1}}{t_2}{\tau}
  }
}
{
  \infer[\textit{(Check Forall)}]{
    \checktype{\Theta;\Gamma}{t}{\tforall{x}{\tnat}{\tau}}
  }{
    \checktype{\Theta;\Gamma,x:\tnat}{\tinst{t}{x}}{\tau}
  }
}

\seprules

\psidebyside{1em}{
  \infer[\textit{(Check Pi)}]{
    \checktype{\Theta;\Gamma}{t}{\tarrow{x}{\tau_1}{\tau_2}}
  }{
    \checktype{\Theta;\Gamma,x:\tau_1}{\tapp{t}{x}}{\tau_2}
  }
}{
  \infer[\textit{(Check Sigma)}]{
    \checktype{\Theta;\Gamma}{t}{\tprod{x}{\tau_1}{\tau_2}}
  }{
    \checktype{\Theta;\Gamma}{\tprojl{t}}{\tau_1} &
    \checktype{\Theta;\Gamma,x:\tau_1,p: \tequal{x}{\tprojl{t}}}{\tprojr{t}}{\tau_2}
  }
}

\seprules

\
{
  \infer[\textit{(Check Refinement)}]{
    \checktype{\Theta;\Gamma}{t}{\trefine{x}{\tau}{b}}
  }{
    \checktype{\Theta;\Gamma}{t}{\tau} &
    {\color{blue}\fbox{$\areequal{\Theta;\Gamma,x:\tau,p:\tequal{x}{t}}{b}{\ttrue}$}}
  }
}

\seprules

\sidebyside
{
  \infer[\textit{(Check Type Abs)}]{
    \checktype{\Theta;\Gamma}{t}{\tpoly{\alpha}{\tau}}
  }{
    \checktype{\Theta,\alpha;\Gamma}{\tinst{t}{\alpha}}{\tau}
  }
}
{
  \infer[\textit{(Check Recursive)}]{
    \checktype{\Theta;\Gamma}{t}{\indexedtype{n_1}{\alpha}{\tau}}
  }{
    \infertype{\Theta;\Gamma}{t}{\indexedtype{n_2}{\alpha}{\tau}} &
    {\color{blue}\fbox{$\areequal{\Theta;\Gamma}{n_1}{n_2}$}}
  }
}

\seprules



\psidebysidebyside{1em}
{
  \infer[\textit{(Check Top 1)}]{
    \checktype{\Theta;\Gamma}{v}{\ttop}
  }{
    \isval{v}
  }
}
{
  \infer[\textit{(Check Top 2)}]{
    \checktype{\Theta;\Gamma}{t}{\ttop}
  }{
    \infertype{\Theta;\Gamma}{t}{\tau}
  }
}
{
  \infer[\textit{(Check Reflexive)}]{
    \checktype{\Theta;\Gamma}{t}{\tau}
  }{
    \infertype{\Theta;\Gamma}{t}{\tau}
  }
}

\caption{
  $\checktype{\Theta;\Gamma}{t}{\tau}$
  checks that term $t$ indeed has type $\tau$ under context $\Theta;\Gamma$.
  When multiple rules are applicable, they are applied from a priority order
  from top to bottom, left to right. The \textit{(Check Forall)} rule can
  be generalized to (non-empty) types other than $\tnat$,
  but we only need it for $\tnat$ (in type $\rectype{\alpha}{\tau}$).
}
\label{fig:bidirectional-check}
\end{figure}


Our inference and checking rules give rise to conditions of the form
$\areequal{\Theta;\Gamma}{t_1}{t_2}$. We call such checks \emph{verification
conditions} (in the rules, they are boxed and appear in blue color). The
$\equiv$ sign is part of the judgment form, and does not
describe a formula. We rely on
an external solver to perform these checks, and assume that when the
verification condition is considered valid by the solver, then: $\forall
\interp, \gamma \in \redvalues{\Theta;\Gamma}{}.
\equivalent{\gamma(\erase{t_1})}{\gamma(\erase{t_2})}$.
This is an equivalent way of saying that
$\redvalues{\tequal{\gamma(\erase{t_1})}{\gamma(\erase{t_2})}}{\interp}$ is non-empty.
Under these conditions, we have the following theorem.

\begin{theorem}[Soundness of the Bidirectional Type-Checker]
\label{th:soundness}
If $\infertype{\Theta;\Gamma}{t}{\tau}$ holds or if
$\checktype{\Theta;\Gamma}{t}{\tau}$ holds, then
$\redd{\Theta;\erase{\Gamma}}{\erase{t}}{\erase{\tau}}$ holds.
\end{theorem}







\subsection{Annotated Terms}
\label{sec:annotated-terms}


\begin{figure}[htb]

\begin{flalign*}
t  & \Coloneqq x\ |\ \tu \hspace{1em} | \hspace{1em}
              \tlambda{x}{\tau}{t} \ |\ \tapp{t}{t} \hspace{1em} | \hspace{1em}
              (t,t) \ |\ \tprojl{t} \ |\ \tprojr{t} \hspace{1em} | \hspace{1em}
              \tinstforall{t}{n} \ \ |\ \\
             & \taleft{\tsum{\tau}{\tau}}{t} \ |\
               \taright{\tsum{\tau}{\tau}}{t} \ |\
               \tsummatch{t}{x \Rightarrow t}{x \Rightarrow t} \hspace{1em} |
               \hspace{1em} \\
             & \ttrue\ |\ \tfalse\ |\ \tite{t}{t}{t}\ |\ \\
             & \tzero\ |\ \tsucc{t}\ |\
             \trec{x \Rightarrow \tau}{t}{t}{(n,y) \Rightarrow t} \ |\
             \tfix{n \Rightarrow \tau}{(n,y) \Rightarrow t} \ |\
             \tmatch{t}{t}{n \Rightarrow t} \ |\
             \\
             & \tfold{\tau}{t} \ |\
              \tunfoldin{t}{x \RA t} \ |\
              \tunfoldposin{t}{x \RA t} \ |\
             \\
             &
              \tabs{\alpha}{t} \ |\ \tinst{t}{\tau} \hspace{0.8em} | \hspace{0.8em}
              \terr{\tau}\ |\ \trefl{t}{t} \ |\ \tulet{x}{t}{t} \ |\ \tsize{t}
\end{flalign*}

\caption{
    Grammar for annotated terms $t$, where $x$, $y$ and $n$ are
    term variables and $\alpha$ is a type variable.
}
\label{fig:terms}
\end{figure}

In order to guide our type-checking algorithm, we require terms to be annotated.
We give in Figure~\ref{fig:terms} the grammar for \emph{annotated terms}. The
term $\tinstforall{t_1}{t_2}$ is used to instantiate a term $t_1$ which has a
type of the form $\tforall{x}{\tau_2}{\tau}$ to a particular term $t_2$ of type
$\tau_2$, in the \textit{(Infer~Forall~Instantiation)} type inference rule of
Figure~\ref{fig:bidirectional-infer}. The term $\tunfoldposin{t_1}{t_2}$ is an
annotated variant of $\tunfoldin{t_1}{t_2}$ (see rules \textit{(Infer~Unfold)}
and \textit{(Infer~Unfold~Positive)} in Figure~\ref{fig:bidirectional-infer}).
We discuss the difference between these rules in
Section~\ref{sec:running-example2-streams}.

The type $\tuletintype{x}{t_2}{\tau}$ represents the type $\tau$ where
the variable $x$ is bound to $t_2$ by using {\tt let}'s in each term that
appears in $\tau$. The formal definition is given in Section~\ref{sec:bidi}.

Annotations such as $\tlambda{x}{\tau}{t}$ or $\tinstforall{t_1}{t_2}$ have no
runtime influence and are erased (respectively to $\tulambda{x}{t}$ and $t_1$).
We write $\erase{t}$ to refer to the \emph{erasure} of $t$, where every
annotation has been erased. The full definition is given in
\iftoggle{arxiv}{Appendix~\ref{app:erase}, Figure~\ref{fig:erase}.} {the long
version of our paper~\citep{SystemFRLongVersion}.}

When a type $\tau$ has annotated terms inside, we write $\erase{\tau}$ to erase
their annotations. For instance $\erase{\trefine{x}{\tau}{b}}$ refers to
$\trefine{x}{\tau}{\erase{b}}$. Moreover, for a context $\Gamma$, we write
$\erase{\Gamma}$ to refer to the context $\Gamma$ where each type $\tau$ has
been replaced by $\erase{\tau}$.


\subsection{Contracts and Measures}
\label{sec:measures}

The syntax we support in our verification tool translates into our
core calculus presented above. In our tool we support
named functions with contracts and measures which are desugared into {\tt fix}
terms. To compare natural numbers and express the fact that
measures decrease, we use
functions `<', `{\tt <=}' and `{\tt ==}' on natural numbers. These functions can be defined
using the recursor {\tt rec}
\iftoggle{arxiv}{(see Appendix \ref{app:lessthan} for definitions).}
{(see the long version of our paper for definitions~\citep{SystemFRLongVersion}).}


Figure~\ref{fig:sugar} shows how, thanks to refinement types, the {\tt fix} term
can encode recursive functions (such as the one given in
Section~\ref{sec:motiv}) that feature user-defined
\emph{pre-} and \emph{post-conditions}
and whose termination arguments relies on a user-defined \emph{measure}
function. The {\tt fix} term shown on the right corresponds to the desugaring of
the recursive function on the left whose contracts are given by the {\bf
require} and {\bf ensuring} keywords, and whose measure is given by the {\bf
decreases} keyword. The contract terms $pre$ and $post$ are such that
{$\hastype{{\tt x}:\tau_1}{pre}{\tbool}$} and {$\hastype{{\tt x}:\tau_1,{\tt
res}:\tau_2}{post}{\tbool}$}, and the measure function $measure$ satisfies
{$\hastype{{\tt x}:\tau_1}{measure}{\tnat}$}. The term $\predsymb$ is a function
of type $\trefine{n}{\tnat}{\tzero < n} \rightarrow \tnat$ that returns the
predecessor of numbers greater than $\tzero$.


\begin{figure}[tbhp]

\begin{minipage}{0.29\textwidth}
\begin{lstlisting}
def f(x: $\tau_1$): $\tau_2$ = {
  require($pre$[x])
  decreases($measure$[x])
  $E$[x, f]
} ensuring {
    res $\RA$ $post$[x, res] }
\end{lstlisting}
\end{minipage}
\begin{minipage}{0.59\textwidth}
\begin{lstlisting}
f $\equiv$
  fix[n $\Rightarrow$ $\tarrow{x}{\trefine{x}{\trefine{x}{\tau_1}{pre}}{measure \leq n}}{\trefine{res}{\tau_2}{post}}$](
    (n, f) $\Rightarrow$ $\lambda{x}:{\trefine{x}{\trefine{x}{\tau_1}{pre}}{measure \leq n}}$.
      $E[x,\tinstforall{f}{\tpred{n}}\tu]$)
\end{lstlisting}
\end{minipage}

\caption{Encoding named function with pre- and post-conditions are given by
  the {\bf require} and {\bf ensuring}, and measure given by
  the {\bf decreases} keyword (left) into a terminating fixpoint recursion (right).
  \label{fig:sugar}}
\end{figure}


We now explain how our type-checking algorithm ensures termination of such a
function. Our type inference rule for {\tt fix} is \textit{(Infer Fix)}. The
side condition $n \notin \fv{\erase{t}}$ ensures that $n$ only appears in type
annotations in $t$, and is not part of the computation. The other check
corresponds to a proof by strong induction (over $n$) that the {\tt fix} term
has type $\tforall{n}{\tnat}{\tau}$. Indeed, we have to check that $t$, the body
of the {\tt fix} term, has type $\tau$ (for some $n: \tnat$), under the
assumption that $y$ (which is the variable representing the recursion) has type
$\subst{\tau}{n}{m}$ for all $m < n$. The `$\tunit \rightarrow$' part of the
type of $y$ corresponds to the fact that the operational semantics of {\tt fix}
replaces variable $y$ by the {\tt fix} term under a lambda (as explained in
Section~\ref{sec:operational}).

The variable $p$ is a witness that the variable $y$ is equal to the {\tt fix}
term (under a lambda). This feature is useful for \emph{body-visible} recursion,
and is explained in Section~\ref{sec:selfaware}.

Back to the encoding presented in Figure~\ref{fig:sugar}, we explain how the
\textit{(Infer Fix)} rule ensures that $measure$ decreases at each recursive
call of function {\tt f}.
Assume that the premise of the \textit{(Infer Fix)} rule holds, and that {\tt f}
is called with some value $v$ of type $\tau_1$, such that $measure[v]$ evaluates
to some (term representing a) natural number $n$. By instantiating the premise
of the \textit{(Infer Fix)} rule for that particular $n$, we get that
$E[x,\tinstforall{f}{\tpred{n}}\tu]$ is well-typed under the condition
that $f$ has type:
\[
    \tforall{m}{\trefine{m}{\tnat}{m < n}}{\tunit\rightarrow
        \tarrow
            {x}
            {\trefine{x}{\trefine{x}{\tau_1}{pre}}{measure \leq m}}
            {\trefine{res}{\tau_2}{post}}
    }
\]

First, in order for $\predsymb$ to be applied to $n$, we have to check that $n$
is non-zero, meaning that the measure of $v$ is strictly positive in the places
where the recursive calls happen. This is ensured by the \textit{(Check
Refinement)} rule for checking refinement types (see
Figure~\ref{fig:bidirectional-check}), which generates a verification condition.



Second, the rule \textit{(Infer Forall Instantiation)} ensures that
$\tinstforall{f}{\tpred{n}}\tu$ takes arguments of type
$\trefine{x}{\trefine{x}{\tau_1}{pre}}{measure \leq \tpred{n}}$. Therefore, if
$f$ is applied recursively to an argument $v'$, the rule \textit{(Check
Refinement)} ensures that $measure[v'] \leq \tpred{n}$ holds. Overall, we get
$measure[v'] \leq \tpred{n} < n = measure[v]$, which ensures that the measures
of arguments always decrease on recursive calls to {\tt f}.

In our implementation, we do not go through the encoding with {\tt fix} and
forall types, but instead directly generate the verification conditions that
correspond to the measure decreasing by using the left-hand-side form of
Figure~\ref{fig:sugar}. Our system also supports mutually recursive functions
(by requiring that the measure decreases for each call to a mutually recursive
function), which can be encoded in the usual way by defining a {\tt fix} term
that returns a tuple of functions.

In the end, if the body of the function is well-typed, the
\textit{(Infer Fix)} rule infers the type:
\[
    \tforall{n}{\tnat}{
        \tarrow
            {x}
            {\trefine{x}{\trefine{x}{\tau_1}{pre}}{m \leq n}}
            {\trefine{res}{\tau_2}{post}}
    }
\]

One should note that this encoding imposes a scoping
restriction on the original program, namely precondition,
postcondition, and measure of a function $f$ cannot contain
calls to $f$. This restriction has not proved limiting in
our experience with benchmarks.

\subsubsection{Lexicographic Orderings}
\label{sec:lexicographic}
Functions whose termination arguments require lexicographic orderings can be
encoded by using two levels of recursions, which is a known technique that shows
expressive power of System T \cite[Section 7.3.2]{ProofsAndTypesBook}.
We review how this encoding works in our system in \iftoggle{arxiv}{Appendix~\ref{app:lexicographic}}{the long version of our paper~\citep{SystemFRLongVersion}} and show an example of Ackerman's function and its simple lexicographic measure. In our implementation, we support lexicographic measures directly.


\subsection{Body-Visible Recursion}

\label{sec:selfaware}


In this section, we give more details about the $\textit{(Infer Fix)}$ and
$\textit{(Infer Rec)}$ typing rules for recursion. They allow \emph{body-visible
recursion} which gives the type-checker access to the definition of a recursive
function while type-checking the body of the recursive function itself.

The first thing to note is that we introduce an equality type containing the
definition (in the type of $p$) in the context, while we do not know yet whether
the body of the recursion is well-typed. Since our equality type is defined (in
Section~\ref{sec:definitions}) for all terms, regardless of whether they are
well-typed, this is perfectly legal. We show in the {\tt merge} example of
Figure~\ref{fig:merge} how body-visible recursion relieves the user from writing
excessive specification annotations.

Assume we want to prove, on paper, that {\tt merge} indeed returns a sorted list
when given two sorted lists {\tt l1} and {\tt l2}, by induction over {\tt
size(l1) + size(l2)}. Consider the first branch of the \tite{}{}{} statement,
where we return {\tt Cons(x, merge(xs,l2))}. By the induction
hypothesis, we know that the recursive call {\tt merge(xs,l2)} is sorted, but
this mere fact is not enough to conclude that {\tt Cons(x,
merge(xs,l2))} is sorted. By unfolding the definition of {\tt isSorted}, we see
that we need in addition to know that {\tt x} is smaller
than the head of the result {\tt merge(xs,l2)}.

Therefore, the property we prove by induction needs to be strengthened by saying
that the head of the result, if non-empty, is equal to the (smallest) head of
one the input lists. From that, we will know by the induction hypothesis that
the head $h$ of {\tt merge(xs,l2)} (if non-empty) is either the head of {\tt xs}
or the head of {\tt l2}. In the first case, we can deduce that {\tt x} is
smaller than $h$ by using the fact that {\tt l1 = Cons(x,xs)} is sorted. In the
second case, we have $h = {\tt y}$, and we know from the condition of the
\tite{}{}{} statement that {\tt x $\leq$ y}. In both cases, we can conclude that
the whole list {\tt Cons(x, merge(xs,l2))} is sorted.

If we are to type-check the program above, and if we only know the return type
of {\tt merge(xs,l2)}, that is {\tt \{ l: \List | isSorted(l) \}}, we will run
into the same problem, and will not be able to conclude that {\tt Cons(x,
merge(xs,l2))} is sorted. In our type system, we get in addition access to the
definition of {\tt merge} while type-checking it, thanks to the $p$ variable of
equality type in the \textit{(Infer~Fix}) rule. By unfolding the definition of
{\tt merge(xs,l2)}, we conclude by case analysis that the head $h$ of {\tt
merge(xs,l2)} (if non-empty) is either the head of {\tt xs} or the head of {\tt
l2} (which is {\tt y}).

Without body-visible recursion, the developer would need to strengthen the
postcondition to:
\begin{center}
\begin{tabular}{c}
\begin{lstlisting}
isEmpty(res) || (!isEmpty(l1) && head(res) == head(l1))
            || (!isEmpty(l2) && head(res) == head(l2))
\end{lstlisting}
\end{tabular}
\end{center}

In Inox, the external solver we use for verification
conditions, definitions of recursive functions are unfolded
automatically. Inox does incremental queries to SMT solvers.
It first sends a query without unfolding at all, then a new
query after unfolding once, and so on until a query succeeds
or a timeout. Thanks to this approach, Inox does not rely on
universal quantifiers to encode recursive functions. This
feature is crucial to have such examples be verified without
user intervention, and is here required to get the bodies of
the calls to {\tt merge} and {\tt isSorted} when verifying
the postcondition of {\tt merge}.


\subsection{Unification in Type Inference}
\label{sec:unification}




\begin{figure}
\[
\begin{array}{l@{}c@{}c@{}c@{}rcl}
  {\unify(E,}\ & \tunit                    & ,  &                    \tunit & ) & \define & \tunit \\
  {\unify(E,}\ & \tbool                    & ,  &                    \tbool & ) & \define & \tbool \\
  {\unify(E,}\ & \tnat                     & ,  &                     \tnat & ) & \define & \tnat  \\
  {\unify(E,}\ & \ttop                     & ,  &                     \ttop & ) & \define & \ttop  \\
  {\unify(E,}\ & \alpha                     & ,  &                     \alpha & ) & \define & \alpha  \\
  {\unify(E,}\ & \tarrow{x}{A_1}{B_1}      & ,  &      \tarrow{x}{A_2}{B_2} & ) & \define & \tarrow{x}{\unify(E, A_1,A_2)}{\unify(E, B_1,B_2)} \\
  {\unify(E,}\ & \tforall{x}{A_1}{B_1}      & ,  &      \tforall{x}{A_2}{B_2} & ) & \define & \tforall{x}{\unify(E, A_1,A_2)}{\unify(E, B_1,B_2)} \\
  {\unify(E,}\ & \tprod{x}{A_1}{B_1}       & ,  &       \tprod{x}{A_2}{B_2} & ) & \define & \tprod{x}{\unify(E, A_1,A_2)}{\unify(E, B_1,B_2)} \\
  {\unify(E,}\ & \indexedtype{n_1}{\alpha}{\tau_1} & , & \indexedtype{n_2}{\alpha}{\tau_2} & ) & \define & \indexedtype{E[n_1,n_2]}{\alpha}{\unify(E, \tau_1,\tau_2)} \\
  {\unify(E,}\ & \trefine{x}{A_1}{p_1}     & ,~ &     \trefine{x}{A_2}{p_2} & ) & \define & \trefine{x}{\unify(E, A_1,A_2)}{E[p_1,p_2]} \\
  {\unify(E,}\ & \tequal{t_{1,1}}{t_{1,2}} & ,  & \tequal{t_{2,1}}{t_{2,2}} & ) & \define & \tequal{E[t_{1,1},t_{2,1}]}{E[t_{1,2},t_{2,2}]} \\
\end{array}
\]

\begin{flalign*}
  {\tt If}\ t_1\  {\tt Then}\  \tau_2\ {\tt Else}\ \tau_3 & \define \unify(\tite{t_1}{\_}{\_}, \tau_2, \tau_3) \\
  {\tt Let}\ x = t\ {\tt in }\ \tau & \define \unify(\tulet{x}{t}{\_}, \tau, \tau) \\
  \tsummatchintype{t}{x \RA \tau_1}{x \RA\tau_2} & \define
    \unify(\tsummatch{t}{x \RA \_}{x \RA \_}, \tau_1, \tau_2) \\
  \tmatchintype{t}{\tau_1}{n \RA\tau_2} & \define
    \unify(\tmatch{t}{\_}{n \RA \_}, \tau_1, \tau_2)
\end{flalign*}

\caption{
  Recursion schema for unification of two types, where $E$ is a context with two holes.
}
  \label{fig:simplification}
\end{figure}

In order to perform type inference, we expect a certain structure on the
inferred types. For example, given term $\tapp{t_1}{t_2}$, we expect term
$t_1$ to have a function type when inferring the type of an application.
Furthermore, type inference must perform least upper bound computation for {\tt
if}, {\tt match} and {\tt either\_match} terms, which adds more complexity to
the system. We handle these considerations in a generalized manner by performing
\emph{type unifications}, defined in Figure~\ref{fig:simplification} by using
the type-level notations {\tt If Then Else}, {\tt Let}, {\tt Either\_Match}, and
{\tt Match}.




We draw attention here to the fact that our type checking and inference
procedures are \emph{syntax directed} and \emph{predictable}, which enables a
natural verification process through term-level hints. The algorithmic nature of
our type checking precludes the verification of certain well-typed programs.
However, our experience has shown that this limitation is largely
inconsequential in practice and is outweighed by the predictable nature of the
algorithm.


\subsection{Hiding Recursive Type Indices}
\label{sec:generalization}






\begin{figure}

\centering
\footnotesize

\
\infer[\textit{(Infer Unfold Gen)}]{
    \infertype{\Theta;\Gamma}{\tunfoldin{t_1}{x \RA t_2}}{\tau'}
}{
    \infertype{\Theta;\Gamma}{t_1}{\rectype{\alpha}{\tau}} &
    \infertype{\Theta;\Gamma,
        x:\rectype{\alpha}{\tau},
        p:\tequal{t_1}{\tufold{x}}
    }{t_2}{\tau'} &
    \tspos{\alpha}{\tau}
}

\seprules

\
\infer[\textit{(Infer Fold Gen)}]{
    \infertype{\Gamma}{\tfold{\rectype{\alpha}{\tau}}{t}}{\rectype{\alpha}{\tau}}
}{
    \checktype{\Gamma}{t}{\subst{\tau}{\alpha}{\rectype{\alpha}{\tau}}} &
    \tspos{\alpha}{\tau}
}

\caption{Folding and unfolding without indices for strictly positive recursive types.}
\label{fig:foldgen}

\end{figure}

Beyond general rules of Figure~\ref{fig:bidirectional-infer}, 
in Figure~\ref{fig:foldgen} we show
two additional rules for `{\tt fold}' and `{\tt unfold in}' that ignore
the indices hidden under the {\tt Rec} type for strictly positive recursive
types (which is the case for the {\tt List} and {\tt Stream} types). We write
$\tspos{\alpha}{\tau}$ when a type variable appears only strictly positively in
type $\tau$, meaning only to the right of $\Pi$ and $\forall$ types (see
\iftoggle{arxiv}{Appendix~\ref{app:polarity}}{the long version of our paper~\citep{SystemFRLongVersion}}
for the precise definition). This enables us, under
some conditions, to {\tt fold} and {\tt unfold} a strictly positive recursive
type without worrying about indices.

Practically, given an element $l$ of type $\PList{\tt X}$ (resp.~$\Stream{\tt
X}$), we can unfold it (with the \textit{(Infer~Unfold~Gen}) rule) to get its
head and its tail of type $\PList{\tt X}$ (resp.~$Unit \rightarrow \Stream{\tt
X}$). Conversely, we use the rule \textit{(Infer~Fold~Gen)} to build a list or a
stream from an element and a tail.

Strict positivity gives us the following key lemma that ensures the soundness of
our rules with respect to our reducibility definition. This lemma states that
when a type variable $\alpha$ appears only strictly positively in $\tau$, then
quantifying with a forall type outside $\tau$ or inside a substitution for
$\alpha$ is the same (as long as we are quantifying over a non-empty type
$\tau_1$). This property is similar to the notions \emph{lim sup-pushable} and
\emph{lim inf-pullable} \cite{DBLP:journals/lmcs/Abel08} and implies the soundness of our rules.



\begin{lemma}
\label{lem:distribute}
Let $\tau$ and $\tforall{x}{\tau_1}{\tau_2}$ be two types. Let $\alpha$ be a
type variable that appears strictly positively in $\tau$. Let $\interp$ be a
type interpretation such that $\redvalues{\tau_1}{\interp}$ is not empty. We
have:
\[
    \redvalues{\subst{\tau}{\alpha}{\tforall{x}{\tau_1}{\tau_2}}}{\interp} =
    \redvalues{\tforall{x}{\tau_1}{\subst{\tau}{\alpha}{\tau_2}}}{\interp}
\]
\end{lemma}

\subsection{Type-Checking Algorithm Examples: Streams}
\label{sec:running-example2-streams}

\subsubsection{Constant Stream}

The {\tt fix} term and associated typing rules can also be used to express the
kind of recursion used to define the streams in Section~\ref{sec:motiv}. We
start by revisiting the constant stream, which in our notations can be written
as an untyped term:
\begin{flalign*}
 & {\tt constant} \define {\tt fix}({\tt constant} \Rightarrow
   \tuabs{.\ \tulambda{x}{\tufold{x,\ \tulambda{u}{
        \tuinst{{\tt constant()}}(x)}}}}
 )
\end{flalign*}
Assume we want to prove, on paper, that for any
value $x$, $\tt \tuinst{{\tt constant}}(x)$ produces a non-diverging stream,
i.e.~a stream which is $n$-non-diverging for every $n \in \Nat$. A natural proof
could be done by induction on $n$, as follows:
\begin{itemize}
\item $(n = 0)$ $\tt \tuinst{{\tt constant}}(x)$ is $0$-non-diverging, meaning
  that it reduces to a value of the form $\tufold{x, v}$ where $x$ and $v$ are
  values. This is clear from the code of {\tt constant}, as this expression
  evaluates in a few steps to $\tufold{x,\ \tulambda{u}{\tuinst{(\lambda
  u'.{\tt constant})()}(x)}}$.
\item $(n = n'+1)$ Assume by induction that $\tt \tuinst{{\tt constant}}(x)$
  is $n'$-non-diverging. By definition of $n$-non-diverging, we get that
  $\tufold{x,\ \tlambda{u}{\tunit}{\tuinst{{\tt constant}}(x)}}$ is
  $n$-non-diverging. Since this term is equivalent to the term to which $\tt
  \tuinst{{\tt constant}}(x)$ evaluates, we conclude that $\tt \tuinst{{\tt
  constant}}(x)$ is $n$-non-diverging as well.
\end{itemize}

Our type system and type-checking algorithm can be used to simulate this proof by
using an annotated version of {\tt constant}:
\[
 {\tt constant} \define {\tt fix}
    [n \Rightarrow \tpoly{X}{X \rightarrow \IStream{X}{n}}]({\tt (n,constant)} \Rightarrow {\tt body(n,constant)})\\
\]
where ${\tt body(n,constant)}$ is a shorthand for
\[
  \tabs{X}{\tlambda{x}{X}{\tfold{\IStream{X}{n}}{x,\ \tlambda{u}{\tunit}{\tinst{\tinstforall{\tt constant}{\tpred{n}}\tu}{X}(x)}}}}
\]

By applying the \textit{(Infer Fix)} rule presented above, we get the type
\[
  \infertype
    {}
    {\tt constant}
    {\tforall{n}{\tnat}{\tpoly{X}{X \rightarrow \IStream{X}{n}}}}.
\]

The \textit{(Infer Fix)} rule of our algorithm generates a check that
corresponds to a (strong) induction that shows that for every $n \in
\redvalues{\tnat}{}$, ${\tt constant}[X](x)$ in $\redvalues{\IStream{X}{n}}{}$
assuming that it is in $\redvalues{\IStream{X}{m}}{}$ for all $m < n$:
\begin{flalign*}
 & {
      n: \tnat,
      {\tt constant}:
      \tforall{m}{\trefine{m}{\tnat}{m < n}}{\tunit \rightarrow \tpoly{X}{X \rightarrow \IStream{X}{n}}}
    } \vdash\\
  & \hspace{5em}{\tt body(n,constant)} \Downarrow
    {\tpoly{X}{X \rightarrow \IStream{X}{n}}}
\end{flalign*}

After applying standard rules related to $\lambda$ and $\Lambda$, our algorithm
will attempt to infer, using the \textit{(Infer Fold)} rule, a type for the
term:
\[
  \tfold{\IStream{X}{n}}{x,\ \tlambda{u}{\tunit}{\tinst{\tinstforall{\tt constant}{\tpred{n}}\tu}{X}(x)}}
\]


In addition to the type-check that $n$ has type $\tnat$, this rule generates two checks
to cover the cases where $n$ is zero or non-zero. These correspond to the
informal proof by induction given above for the non-divergence of {\tt
constant}. The first check reduces (after applying some straightforward rules)
to checking that $\tlambda{u}{\tunit}{\tinst{\tinstforall{\tt
constant}{\tpred{n}}\tu}{X}(x)}$ has type $\ttop$ (remember that
$\basetype{\alpha}{X \times (\tunit \rightarrow \alpha)} = X \times \ttop$),
which goes through easily thanks to the rule \textit{(Check Top 1)}.

The second check amounts to checking that $\tinstforall{\tt
constant}{\tpred{n}}\tu{\tt [X](x)}$ has type $\IStream{X}{n'}$ under the assumption that
$\tequal{n}{\tsucc{n'}}$. By the context we know that $\tinstforall{\tt
constant}{\tpred{n}}\tu{\tt [X](x)}$ has type $\IStream{X}{\tpred{n}}$. Since $n'$ and
$\tpred{n}$ are equivalent, we can use the rule \textit{(Check Recursive)}
(given in Figure~\ref{fig:bidirectional-check}) to convert between the two
types.

Attempting to type-check the {\tt badConstant} example
from Section~\ref{sec:motiv}:
\[
 {\tt fix}
    [n \Rightarrow \tpoly{X}{X \rightarrow \IStream{X}{n}}]({\tt (n,badConstant)} \Rightarrow {\tt badBody(n,badConstant)})\\
\]
where ${\tt badBody(n,badConstant)}$ stands for:
\[
  \tabs{X}{\tlambda{x}{X}{\tfold{\IStream{X}{n}}{x,\
    \tlambda{u}{\tunit}{\tinst{\tinstforall{\tt badConstant}{\tpred{n}}\tu}{X}(x)}.{\tt tail()}}}}
\]
will lead to an error in the second check (corresponding to
the inductive case), since the extra call to {\tt tail()}
decreases the index of the stream by one.




\subsubsection{ZipWith Function on Streams}

We now revisit the {\tt zipWith} function from Figure~\ref{fig:zipwith}. We said
in Section~\ref{sec:motiv} that for every $n \in \Nat$, when {\tt s1} and {\tt
s2} are $n$-non-diverging streams, then so is {\tt zipWith(f,s1,s2)}. On paper,
we can check by induction over $n \in \Nat$ that when {\tt s1} and {\tt s2} are
$n$-non-diverging streams and $f$ is terminating, then {\tt zipWith(f,s1,s2)}
(as written in Figure~\ref{fig:zipwith}, and ignoring type annotations and the erasable
parameter for the moment) is also an $n$-non-diverging stream.
\begin{itemize}
\item $(n = 0)$
  It is indeed the case that we can access the head of
  {\tt zipWith(f,s1,s2)} as long as we can access the heads {\tt s1}
  and {\tt s2} (and as long as $f$ terminates).
\item $(n = n'+1)$ Let {\tt s1} and {\tt s2} be two $n$-non-diverging streams.
  By definition of non-diverging, we know that \stail{\tt s1} and \stail{\tt s2}
  are $n'$-non-diverging. By induction hypothesis,\\
  {\tt zipWith(f,\stail{s1},\stail{s2})} is $n'$-non-diverging as well. This
  means that {\tt zipWith(f,s1,s2)} is $n$-non-diverging, which concludes the
  proof.
\end{itemize}


\begin{figure}

\scalebox{0.87}{\parbox{\linewidth}{
\begin{flalign*}
 & {\tt zipWith} \triangleq {\tt fix}
    [n \Rightarrow \tpoly{X,Y,Z}{(X \rightarrow Y \rightarrow Z)
      \rightarrow \IStream{\tt X}{n}
      \rightarrow \IStream{\tt Y}{n}
      \rightarrow \IStream{\tt Z}{n}}](\\
& \hspace{1em}{\tt (n,zipWith)} \Rightarrow \\
 &   \hspace{2em}
      \Lambda X,Y,Z.\
        \lambda{\tt f: X \rightarrow Y \rightarrow Z}.\
        \lambda{\tt s1: \IStream{\tt X}{n}}.\
        \lambda{\tt s2: \IStream{\tt Y}{n}}.\\
 & \hspace{3em}
    {\tt fold}[\IStream{\tt Z}{n}](
    \\
 & \hspace{4em}
    {\tt unfold\ s1\ in\ xs1\ } \RA \\
 & \hspace{4em}
    {\tt unfold\ s2\ in\ xs2\ } \RA
    \\
 & \hspace{5em}
    {\tt f}\ (\tprojl{\tt xs1})\ (\tprojl{\tt xs2}),
    \\
 & \hspace{4em}
    \lambda u: \tunit. \\
 & \hspace{5em}
    {\tt unfold\_pos\ s1\ in\ xs1\ } \RA \\
 & \hspace{5em}
    {\tt unfold\_pos\ s2\ in\ xs2\ } \RA
    \\
 & \hspace{6em}
    {\tinst{\tinst{\tinst{\tinstforall{\tt zipWith}{\tpred{n}}()}{X}}{Y}}{Z}\
      ((\tprojr{\tt xs1})\ \tu)\
      ((\tprojr{\tt xs2})\ \tu)
    } ))
\end{flalign*}
}}

\caption{An annotated term of our calculus to define {\tt zipWith}.}
\label{fig:zipwithterm}

\end{figure}

Accordingly, our type-checking algorithm is able to infer the type
\[
  \tforall{n}{\tnat}{\tpoly{X,Y,Z}{(X \rightarrow Y \rightarrow Z)
      \rightarrow \IStream{\tt X}{n}
      \rightarrow \IStream{\tt Y}{n}
      \rightarrow \IStream{\tt Z}{n}}}.
\] for the annotated {\tt zipWith} term given in Figure~\ref{fig:zipwithterm}.
We use the term `{\tt unfold in}' to access the heads (with $\tprojl{}$). For
the tails, we use instead `{\tt unfold\_pos in}'. To understand the difference,
we must first look at the \textit{(Infer Fold)} rule for the {\tt fold} term,
which will generate two subgoals, one with $n = 0$ and one with $n > 0$. In the
subgoal with $n = 0$, we must check that the lambda term `$\lambda u: Unit.
[...]$' has type $\ttop$, which goes through directly thanks to
\textit{(Check~Top~1)}. Therefore, when we are type-checking the body of the
lambda, we know that $n > 0$. We can thus access the tails using
\textit{(Infer~Unfold~Positive)}, a variant of \textit{(Infer~Unfold)} that
discards a subgoal with $n = 0$ but requires proving $n > 0$ instead.

\subsubsection{Fibonacci Stream}

We now consider the Fibonacci stream function from Figure~\ref{fig:fibonacci}.
We can prove by induction that for all $n \in \Nat$, {\tt fib} is an
$n$-non-diverging stream. In our view, this corresponds to writing {\tt fib}
as in Figure~\ref{fig:fibonacciterm}.


\begin{figure}

\scalebox{0.87}{\parbox{\linewidth}{
\begin{flalign*}
 & {\tt fib} \define {\tt fix}
    [n \Rightarrow \IStream{\tnat}{n}]({\tt (n,fib)} \Rightarrow \\
 & \hspace{2em}
    {\tt fold}[\IStream{\tnat}{n}](0, \lambda u: \tunit.\ \\
 & \hspace{3em}
    {\tt fold}[\IStream{\tnat}{\tpred{n}}](1, \lambda u: \tunit.\
    \\
 & \hspace{4em}
    {\tt unfold\_pos\ \tinstforall{fib}{\tpred{n}}()\ in\ xfib\ \RA}
    \\
 & \hspace{4em}
    {\tt inst(zipWith,\tpred{\tpred{n}})[\tnat][\tnat][\tnat]}\\
 &   \hspace{5em}{\tt plus} \ \ 
    {\tt (\tinstforall{fib}{\tpred{\tpred{n}}}())} \ \
    {\tt ((\tprojr{\tt xfib})\ \tu)})))
\end{flalign*}
}}

\caption{An annotated term of our calculus to define the Fibonacci stream.}
\label{fig:fibonacciterm}

\end{figure}


When type-checking {\tt fib}, we use the type of {\tt zipWith} and
instantiate it to \tpred{\tpred{n}}, and get that:
{\tt inst(zipWith,\tpred{\tpred{n}})} has type:

\scalebox{0.90}{\parbox{\linewidth}{
\[
  \tpoly{X,Y,Z}{(X \rightarrow Y \rightarrow Z)
      \rightarrow \IStream{\tt X}{\tpred{\tpred{n}}}
      \rightarrow \IStream{\tt Y}{\tpred{\tpred{n}}}
      \rightarrow \IStream{\tt Z}{\tpred{\tpred{n}}}
  }
\]
}}

This is where we use the fact that {\tt zipWith} returns an
$n$-non-diverging stream when given $n$-non-diverging streams.




\subsection{Verification of Properties}
\label{sec:running-example2-verif}

Finally, we show how to verify the property mentioned in
Section~\ref{sec:motiv}: that zipping two streams {\tt s1} and {\tt s2} with
the function $\tulambda{x}{\tulambda{y}{x}}$ returns a stream equivalent to
{\tt s1}.

We define a function {\tt zipWith\_fst} with the following type:

\scalebox{0.93}{\parbox{\linewidth}{
\[
  \tarrow{n}{\tnat}{
  \tpoly
    {\tt X}
    {\tarrow
      {\tt s1}
      {\Stream{\tt X}}
      {\tarrow
        {\tt s2}
        {\Stream{\tt X}}
          {\tequal
            {{\tt nth}\ n\ {\tt s1}}
            {{\tt nth}\ n\ (zipWith(\tulambda{x}{\tulambda{y}{x}}, \tt s1, s2))}}
      }
    }
  }
\]
}}


\begin{figure}

\footnotesize
\begin{flalign*}
{\tt zipWith\_fst} \equiv\ &\lambda n: \tnat.\ {\tt rec}[n \Rightarrow \tau](n, \\
  &\hspace{1em}
    \Lambda{\tt X}.\ \lambda {\tt s1}: \Stream{X}.\ \lambda {\tt s2}: \Stream{X}.\\
  &\hspace{2em}
    {\trefl{{\tt nth}\ 0\ {\tt s1}}{{\tt nth}\ 0\ (zipWith(\tulambda{x}{\tulambda{y}{x}}, \tt s1, s2))}}, \\
  & \qquad (n,{\tt zipWith\_fst}) \Rightarrow \\
  &\hspace{2em}
    \Lambda{\tt X}.\ \lambda {\tt s1}: \Stream{X}.\ \lambda {\tt s2}: \Stream{X}.\\
 & \hspace{3em}
    {\tt unfold\ s1\ in\ xs1\ } \RA \\
 & \hspace{3em}
    {\tt unfold\ s2\ in\ xs2\ } \RA
    \\
  &\hspace{3em}
   {\tt let\ pr:
      \tequal
        {{\tt nth}\ n\ (\tprojr{\tt xs1})}
        {{\tt nth}\ n\ (zipWith(\tulambda{x}{\tulambda{y}{x}}, \tprojr{\tt xs1}, \tprojr{\tt xs2}))}}
      =\\
  &\hspace{4em} \tinst{\tt zipWith\_fst\tu}{\tt X}\ (\tprojr{\tt xs1})\ (\tprojr{\tt xs2})\ {\tt in}
      \\
  &\hspace{3em}
    {\trefl{{\tt nth}\ (\tsucc{n})\ {\tt s1}}{{\tt nth}\ (\tsucc{n})\ (zipWith(\tulambda{x}{\tulambda{y}{x}}, \tt s1, s2))}})    
\end{flalign*}

\caption{An annotated term of our calculus to prove that zipping two streams
with $\protect\lambda${}x.$\protect\lambda${}y.x returns a stream equivalent to the first.}
\label{fig:zipfstterm}

\end{figure}

For short, we denote this type as $\tarrow{n}{\tnat}{\tau}$. The code of {\tt
zipWith\_fst} is given in Figure~\ref{fig:zipfstterm} and makes use of the
recursor {\tt rec}. Our \textit{(Infer~Rec)} rule for inferring the type of {\tt
rec} is similar to {\tt fix}. The difference is that {\tt rec} allows references
to $n$ in the code, so the computation is allowed to depend on $n$. Moreover,
{\tt rec} does not use $\forall$ types. Finally, {\tt rec} is analogous to a
simple induction, while {\tt fix} is analogous to a strong induction.


In the {\tt zipWith\_fst} example, applying this rule corresponds to a proof
by induction over $n$, that for any two streams {\tt s1} and {\tt s2}, the $n$th
element of {\tt s1} is equal to the nth element of
$zipWith(\tulambda{x}{\tulambda{y}{x}}, \tt s1, s2)$.

In the base case, to type-check the term
${\trefl{{\tt nth}\ 0\ {\tt s1}}{{\tt nth}\ 0\ (zipWith(\tulambda{x}{\tulambda{y}{x}}, \tt s1, s2))}}$,
we must make sure that the following equality holds:
$
  \tequal{{\tt nth}\ 0\ {\tt s1}}{{\tt nth}\ 0\ (zipWith(\tulambda{x}{\tulambda{y}{x}}, \tt s1, s2))}.
$

The corresponding type inference rule is \textit{(Infer Refl)}, which generates
a verification condition.

In the inductive case, we explicitly instantiate our inductive hypothesis on the
tails of {\tt s1} and {\tt s2} by using the {\tt let} binding on variable {\tt
pr}. The equality given by the type of {\tt pr} is then sufficient to prove
what we wanted:
$
  \tequal{{\tt nth}\ (\tsucc{n})\ {\tt s1}}{{\tt nth}\ (\tsucc{n})\ (zipWith(\tulambda{x}{\tulambda{y}{x}}, \tt s1, s2))}.
$

  This example also illustrates that 
developers can express the desired versions of what others might call
extensional equality (here: same results when calling {\tt nth}) using $\Pi$ types and equi-reducibility $\equiv$; our type system does not impose any preferred form of extensional equality.


\section{Formalization in the Coq Proof Assistant}
\label{sec:coq}


We here give more details about our formalization of Theorem~\ref{th:soundness}
in Coq 8.9.1 (including the rules from Section~\ref{sec:generalization}). Our
proofs are available from \\
\centerline{\url{https://github.com/epfl-lara/SystemFR/tree/oopsla2019}}
We represent terms and types using a \emph{locally nameless
representation}~\cite{DBLP:journals/jar/Chargueraud12}, where free variables are
named, and where local variables are bound using De Bruijn indices. Using this
representation, lambdas and other binders can be seen as terms with holes, which
can be filled with other terms (typically values). We use the Coq Equations
library~\cite{DBLP:conf/itp/Sozeau10} to define the reducibility logical
relation $\redvalues{}{}$. This library facilitates the use of functions which
are defined recursively based on a well-founded measure.

We give an overview of our files (containing around $20$k lines of code):
\begin{itemize}
\item \href{https://github.com/epfl-lara/SystemFR/blob/oopsla2019/Trees.v}{\tt Trees.v} contains the definitions of types and terms,
\item \href{https://github.com/epfl-lara/SystemFR/blob/oopsla2019/Typing.v}{\tt Typing.v} gives all typing rules (containing rules from the paper and more),
\item \href{https://github.com/epfl-lara/SystemFR/blob/oopsla2019/SmallStep.v}{\tt SmallStep.v} contains the operational semantics of the language,
\item \href{https://github.com/epfl-lara/SystemFR/blob/oopsla2019/ReducibilityDefinition.v}{\tt ReducibilityDefinition.v} contains the definition of reducibility,
\item The {\tt Reducibility*.v} files contain lemmas for the soundness of the rules from {\tt Typing.v},
\item \href{https://github.com/epfl-lara/SystemFR/blob/oopsla2019/Reducibility.v}{\tt Reducibility.v} contains the proof that all typing rules from \href{https://github.com/epfl-lara/SystemFR/blob/oopsla2019/Typing.v}{\tt Typing.v} are sound with
  respect to the reducibility definition
  (which implies that Theorem~\ref{th:soundness} holds).
\end{itemize}




\subsection{Extensions of Formalization}
\label{sec:ext}


\begin{figure}
  \newcommand{\mnl}{\\[0.6ex]}
  
\begin{equation*}\begin{array}{rcl@{\qquad\qquad}lcl}
    \redvalues{\tbot}{\interp} &\define& \set{} &
    \redvalues{\tsingleton{t}}{\interp} &\define& \set{v \in \Val \ |\ t \smallstep^* v} \mnl
    \redvalues{\tintersect{\tau_1}{\tau_2}}{\interp} &\define& \redvalues{\tau_1}{\interp} \cap \redvalues{\tau_2}{\interp} &
    \redvalues{\tunion{\tau_1}{\tau_2}}{\interp}  &\define&  \redvalues{\tau_1}{\interp} \cup \redvalues{\tau_2}{\interp} \mnl
    \redvalues{\texists{x}{\tau_1}{\tau_2}}{\interp}  &\define&
    \multicolumn{4}{l}{
  \set{b\ |\
    \exists a \in \redvalues{\tau_1}{\interp}.\ b \in \redvalues{\subst{\tau_2}{x}{a}}{\interp}}
     \hspace*{2em}\mbox{(existential)}} \\
    \redvalues{\ttyperefine{x}{\tau_1}{\tau_2}}{\interp}  &\define&
    \multicolumn{4}{l}{
      \set{a \in \redvalues{\tau_1}{\interp}\ |\ \exists b \in \redvalues{\subst{\tau_2}{x}{a}}{\interp}}
      \hspace*{3em}\mbox{(refinement by type)}
    } \\
    \redvalues{\tuletintype{x}{t}{\tau}}{\interp}  &\define&
    \multicolumn{4}{l}{
  \set{v\ |\ \exists a \in \Val.\ t \smallstep^* a \land v \in \redvalues{\subst{\tau}{x}{a}}{\interp}}
    }
\end{array}\end{equation*}

\caption{Definition of reducibility for intersection, union, singleton
types, and refinement by types.}
\label{fig:extreducibility}

\end{figure}


{
\footnotesize
\begin{figure}
\begin{equation*}\begin{array}{rcl}
\tujudge{t}{\tau} &\define& \texists{x}{\tau}{\tequal{x}{t}} \hspace*{8em} \mbox{(judgement-as-type)} \\
\tsetcomp{\tau} &\define& \ttyperefine{x}{\top}{(\tujudge{x}{\tau}) \to \bot} \hspace*{3.8em} \mbox{(set complement)} \\
\timage{f}{\tau} &\define& \ttyperefine{y}{\top}
       {\texists{x}{\tau}{\tequal{f x}{y}}}
       \hspace*{2.5em} \mbox{(type image)} \\
\tinvimage{f}{\tau} &\define& \ttyperefine{x}{\top}{\texists{y}{\tau}{\tequal{f x}{y}}}
       \hspace*{2.5em} \mbox{(type pre-image)} \\
\tite{b}{\tau_1}{\tau_2} &\define&
  \tunion{\trefine{x}{\tau_1}{b}}{\trefine{x}{\tau_2}{{\tt not}(b)}}  \\
\ttite{\tau}{\tau_1}{\tau_2} &\define&
\tunion{\ttyperefine{x}{\tau_1}{\tau}}{\ttyperefine{x}{\tau_2}{\tau \to \bot}}  \\
\tmatch{t}{\tau_1}{x \Rightarrow \tau_2} &\define&
  \tunion
    {\ttyperefine{x}{\tau_1}{\tequal{t}{\tzero}}}
    {\tuletintype{x}{\tpred{t}}{\tau_2}}  \\
\tsummatch{t}{x \Rightarrow \tau_1}{x \Rightarrow \tau_2} &\define&
    (\tuletintype{x}{{\tt unfold\_left}(t)}{\tau_1})\ \ \ \cup \\
& & (\tuletintype{x}{{\tt unfold\_right}(t)}{\tau_2})
\end{array}\end{equation*}
\begin{equation*}\begin{array}{rcl}
\mbox{where} \hspace*{8em}
{\tt not}(b) &\define& \tite{b}{\tfalse}{\ttrue} \\
{\tt unfold\_left}(t) &\define& \tsummatch{t}{x \Rightarrow x}{\tuerr} \\
{\tt unfold\_right}(t) &\define& \tsummatch{t}{\tuerr}{x \Rightarrow x}
\end{array}\end{equation*}

\caption{Syntactic sugar for further expressive types}
\label{fig:typesugar}

\end{figure}
}

Our approach to formalization proved flexible and allowed us to extend the system with additional
types. In addition to formalizing the types and rules for our type checking algorithm
(figures~\ref{fig:bidirectional-check} and \ref{fig:bidirectional-check}),
our Coq formalization also defines 
(Figure~\ref{fig:extreducibility})
reducibility for the bottom type, singleton
types, union and intersection types, and an existential type more abstract than dependent pair (and dual to
$\tforall{x}{\tau_1}{\tau_2}$). We also formalized \emph{refinement by type}
$\ttyperefine{x}{\tau_1}{\tau_2}$, which generalizes refinements from computable terminating
terms of our language
to the condition of non-emptiness of arbitrary types in our system. Existential type and refinement by type
can be viewed as second and first projections of dependent pairs.
Together with existential and universal types, refinement by type allows
us to refine types by quantified propositions, providing support for quantified preconditions and postconditions.
Some of the type forms we present can be expressed using others;
the set of type forms is not minimal.

We define further syntactic sugar, including precise {\tt if then else}
and {\tt match} types (Figure~\ref{fig:typesugar}), as an alternative to
the $\titeintype{t_1}{t_2}{t_3}$ that relies on simplifications of Figure~\ref{fig:simplification}.
We rely on {\tt Let} type of Figure~\ref{fig:extreducibility}, and we
use the ill-typed $\tuerr$ term in the computations of (respectively) \predsymb,
{\tt unfold\_left} and {\tt unfold\_right}, to make sure that the denotations of
the {\tt Let} types are empty when $t$ does not reduce (respectively) to a
strictly positive natural number, $\tuleft{t'}$, or $\turight{t'}$ for some $t'$.
We proved soundness of typing rules that show that these types behave as
expected. For instance, an {\tt if then else} expression can be assigned an
{\tt if then else} type of the branches.

  These types can be used to encode type level computation reminiscent to one present in Scala
{\sf HList}\footnote{\url{https://github.com/milessabin/shapeless/blob/master/core/src/main/scala/shapeless/hlists.scala}}, as illustrated by the following example.
Given a list $L: \PList{Bool}$ and types $\tau_1$ and $\tau_2$, we define a
type ${\tt TypeFromList}(L,\tau_1,\tau_2)$ that represents nested pairs from 
$\tau_1$ and $\tau_2$ according to the boolean values 
($\tfalse$ for $\tau_1$ and $\ttrue$ for $\tau_2$) and with nesting given by the length of $L$.
For example, when $L$ is the list $\ttrue, \tfalse, \ttrue$ then 
${\tt TypeFromList}(L, \tnat, \tnat \to \tnat)$ 
represents nested tuples $w$ of the form $(a_0, (a_1, (a_2,
\ldots)))$ where $a_0:  \tnat$, $a_1: \tnat \to \tnat$, $a_2: \tnat$.
{
\footnotesize
\begin{flalign*}
 & {\tt TypeFromList}(L,\tau_1,\tau_2) \define \dobracket w: \ttop\ |\ \forall n: \{ n:\tnat\ |\ n < {\tt size}(L) \}.\ 
  \tujudge{{\tt get}(w,n)}{({\tt if\ } {\tt nth}(L,n)\ {\tt then\ } \tau_1 {\tt\ else\ } \tau_2)} \dcbracket
\end{flalign*}
}
The function ${\tt nth}(L,n)$ returns the $n$th element of
the list $L$ (defined as a recursive data type).
The function ${\tt get}(w, n)$ is also a term-level function that operates on
nested pairs $w$ such that e.g.\ ${\tt get}(w, 2)$ expands to
$\pi_1(\pi_2(\pi_2(w)))$.
The type
$\tequal{x}{{\tt get}(w,n)}$ in the expansion of refinement by judgement-as-type $\tujudge{\_}{\_}$
(Figure~\ref{fig:typesugar})
does not require type checking the ${\tt get}$ function; its semantics
(see $\approx$ in Figure~\ref{fig:reducibility}) is
that ${\tt get}(w,n)$ reduces to $x$ in the (untyped) operational semantics.

Furthermore, we prove soundness of multiple rules (such as congruence rules)
for establishing the equality judgment $\areequal{\Theta;\Gamma}{t_1}{t_2}$
described in Section~\ref{sec:automated-verifier}. We also prove soundness of a
rule that unfolds the definition of a recursive function in the context, which
is what is required by the solver in Section~\ref{sec:selfaware} to unfold
the definition of {\tt merge}.
Such rules are a step towards justifying
not only the verification condition generation but also verification condition
solving.

\section{Implementation and Evaluation}
\label{sec:implementation}

We have implemented our bidirectional type checking procedure
by writing an alternative verification-condition generator for
Stainless\footnote{\url{https://github.com/epfl-lara/stainless}}
\cite{Stainless}.
The code is merged into master and available in e.g. \href{https://github.com/epfl-lara/stainless/tree/v0.4.0}{release version 0.4.0}.
Its functionality can be invoked with the \verb|--type-checker|
command line option in both scalac 2.12 and Dotty front end pipeline for Stainless.
Thanks to these frontends, which provide a form of type inference,
the type annotation burden is lessened for the user. Precise
types such as indexed recursive types (which
are not supported out of the box by these frontends) still
need to be annotated manually.

The implementation was evaluated on benchmarks shown in
Figure~\ref{fig:evaluation} totalling 14k LoC,
collected from existing Stainless test suites and case
studies. The proof/code ratio depends on the properties
being proven. It can be 0 or close to 0 as in the {\tt
merge} example of Section~\ref{sec:motiv}, or higher than 1
when writing lemmas for proving detailed specifications
(as in the \href{https://github.com/epfl-lara/stainless/blob/v0.4.0/frontends/benchmarks/typechecker/valid/fp-principles/huffman-coding/Huffman.scala}{Huffman coding example}).
The benchmarks reside in
the \verb|frontends/benchmarks/typechecker/valid| directory of Stainless.
The streams benchmark relies on more expressive annotations of the Dotty compiler that we use
to write recursive types, and is available in
\verb|frontends/benchmarks/dotty-specific/typechecker|. We use the following syntax to represent
indexed types:
\begin{center}
\begin{tabular}{|c|c|c|} \hline
& notation in this paper & encoding in Dotty front-end \\ \hline
type & \lstinline|$\IStream{X}{n}$|        & \lstinline|$\Stream{X}$ @indexedAt(n)| \\ \hline
constructor & \lstinline|$\tfold{\IStream{X}{n}}{(hd,tl)}$| & \lstinline|indexedAt(n, $\Stream{X}$(hd,tl))| \\ \hline
\end{tabular}
\end{center}
The suite of all benchmarks verifies in $\sim$6 minutes in total
when using implementation based on our type checker.
The table shows the number of lines of
codes within each benchmark, and the number of verification conditions which
were generated by the type checker and verified by the SMT-backed Inox solver.
The time given includes the time for generating the verification conditions and
solving them, but not the parsing and compilation which is done by the Scala
compiler {\tt scalac} (which we use for some initial type inference and to
obtain a tree representation of the program), nor the transformations which are
internal to Stainless and which happen before using the type-checker.
To understand the impact of these time measurements, note that 
11\% of time of total verification time is spent in
the bidirectional type-checking algorithm that generates verification
conditions and that is the focus of this paper, 
66\% in checking the verification conditions using a SMT-backed solver Inox for
recursive and higher-order functions,
13\% for the parsing, name resolution and most of the type checking pipeline
of Scala's {\tt scalac} compiler, and
10\% for extraction of {\tt scalac} abstract syntax trees to Stainless trees.
All times were
measured on a Lenovo X1 Carbon laptop with an Intel i7-7500U and 16GB RAM.




\begin{figure}
\scriptsize
\setlength\tabcolsep{1pt}

\renewcommand{\arraystretch}{0.90}

\scalebox{1}{
\begin{minipage}[t]{0.33\textwidth}
\begin{tabular}[t]{lrrr}
  Name                  &  LoC & VCs & Time (s) \\ \hline
AbstractRefinementMap	& 14	& 8	& 0.80 \\
Acc	& 26	& 1	& 0.14 \\
Ackermann	& 12	& 12	& 0.29 \\
AliasPartial	& 19	& 2	& 0.10 \\
AmortizedQueue	& 155	& 64	& 3.92 \\
AnyDown	& 24	& 3	& 0.21 \\
AssociativeFold	& 102	& 36	& 2.14 \\
AssociativeList	& 66	& 28	& 0.80 \\
BalancedParentheses	& 410	& 214	& 25.76 \\
BasicReal	& 28	& 6	& 0.18 \\
Basics	& 442	& 137	& 3.60 \\
BestRealTypes	& 27	& 1	& 0.14 \\
BigIntMonoidLaws	& 38	& 9	& 1.96 \\
BinarySearch	& 18	& 7	& 6.61 \\
BinomialHeap	& 189	& 22	& 0.95 \\
BitsTricks	& 103	& 35	& 0.72 \\
BooleanOps	& 50	& 12	& 0.41 \\
BottomUpMergeSort	& 127	& 62	& 2.80 \\
Bytes	& 33	& 7	& 0.12 \\
ChurchNum	& 21	& 0	& 0.03 \\
Client	& 20	& 1	& 0.09 \\
ConcRope & 493 & 254 & 49.21 \\
ConcTree & 335 & 180 & 9.44 \\
CountTowardsZero & 16 & 7 & 0.08 \\
Countable & 29 & 17 & 0.49 \\
Deque & 241 & 80 & 3.73 \\
DivisionByZero & 21 & 3 & 0.06 \\
EffectfulPost & 18 & 1 & 0.08 \\
EitherLaws & 34 & 5 & 0.15 \\
Factorial & 47 & 23 & 0.60 \\
Fibonacci & 19 & 7 & 0.10 \\
FiniteSort & 61 & 14 & 1.02 \\
FiniteStreams & 42 & 13 & 0.30 \\
FlatMap & 60 & 20 & 0.66 \\
Foldr & 19 & 3 & 0.09 \\
FoolProofAdder & 13 & 1 & 0.04 \\
Formulas & 55 & 30 & 2.30 \\
\end{tabular}
\end{minipage}
\begin{minipage}[t]{0.34\textwidth}
\begin{tabular}[t]{lrrr}
Name & LoC & VCs & Time \\ \hline
FunSets & 134 & 19 & 1.11 \\
GodelNumbering & 486 & 389 & 84.00 \\
HOInvocations & 19 & 2 & 0.07 \\
HammingMemoized & 67 & 46 & 2.60 \\
Heaps & 166 & 90 & 3.59 \\
Huffman & 520 & 228 & 9.52 \\
IgnoredField & 48 & 2 & 0.07 \\
IndirectHO & 19 & 2 & 0.04 \\
Induction & 415 & 150 & 3.04 \\
InsertionSort & 85 & 31 & 0.46 \\
IntSet & 65 & 29 & 0.51 \\
IntSetProp & 66 & 21 & 0.57 \\
IntSetUnit & 60 & 23 & 0.51 \\
Justify & 79 & 46 & 0.90 \\
Knapsack & 73 & 23 & 0.30 \\
LazyNumericalRep & 158 & 29 & 1.15 \\
LazySelectionSort & 69 & 21 & 0.40 \\
LeftPad & 65 & 25 & 1.10 \\
List & 1006 & 529 & 16.07 \\
ListMonad & 86 & 29 & 0.97 \\
ListMonoidLaws & 47 & 22 & 0.46 \\
ListWithSize & 169 & 105 & 1.56 \\
LiteralMaps & 25 & 0 & 0.05 \\
Longs & 36 & 8 & 3.53 \\
Map & 21 & 3 & 0.11 \\
MapGetOrElse2 & 19 & 2 & 0.14 \\
MapGetPlus & 20 & 10 & 0.41 \\
McCarthy91 & 24 & 8 & 0.17 \\
Mean & 13 & 2 & 0.27 \\
MergeSorts & 301 & 167 & 4.51 \\ 
Methods & 32 & 13 & 0.28 \\
MicroTests & 2246 & 505 & 15.59 \\ 
Monoid & 37 & 12 & 0.49 \\
MoreExtendedEuclidGCD & 74 & 39 & 7.18 \\
MySet & 18 & 2 & 0.03 \\
NNF & 135 & 177 & 17.77 \\
NNFSimple & 107 & 103 & 5.27 \\
\end{tabular}
\end{minipage}
\begin{minipage}[t]{0.33\textwidth}
\begin{tabular}[t]{lrrr}
Name & LoC & VCs & Time \\ \hline
NatMonoidLaws & 75 & 39 & 0.71 \\
NaturalBuiltin & 25 & 8 & 0.11 \\
NestedLoop & 19 & 4 & 0.08 \\
NotEquals & 22 & 1 & 0.09 \\
Numeric1 & 18 & 4 & 0.06 \\
OddEven & 77 & 24 & 0.39 \\ 
OptionMonad & 47 & 9 & 0.26 \\
Overrides & 25 & 9 & 0.24 \\
PackratParsing & 137 & 31 & 1.06 \\
ParBalance & 250 & 88 & 1.60 \\
PartialCompiler & 66 & 38 & 10.67 \\
PartialKVTrace & 77 & 19 & 3.69 \\
Patterns & 28 & 10 & 0.78 \\
Peano & 36 & 12 & 0.28 \\
PositiveMap & 43 & 15 & 0.41 \\
PreInSpecs & 27 & 6 & 0.09 \\
PropositionalLogic & 89 & 114 & 16.04 \\
Queue & 28 & 9 & 0.25 \\
QuickSorts & 220 & 131 & 5.65 \\ 
ReachabilityChecker & 561 & 282 & 15.34 \\
RealTimeQueue & 77 & 9 & 0.33 \\
RedBlackTree & 115 & 49 & 1.52 \\
SearchLinkedList & 59 & 22 & 0.27 \\
Shorts & 33 & 7 & 0.10 \\
SimpInterpret & 71 & 18 & 1.18 \\
StableSorter & 129 & 52 & 1.51 \\
Streams & 87 & 147 & 4.41 \\
Termination\_passing1 & 37 & 10 & 0.16 \\
Theorem & 28 & 1 & 0.05 \\
ToChurch & 27 & 4 & 0.07 \\
Trees1 & 29 & 11 & 0.43 \\
TweetSet & 494 & 253 & 11.91 \\
UpDown & 49 & 11 & 0.25 \\
Viterbi & 120 & 19 & 0.33 \\
XPlus2N & 19 & 4 & 0.07 \\
example & 95 & 70 & 1.74 \\
recfun & 66 & 38 & 0.89 \\
\end{tabular}
\end{minipage}
}

\leftline{\rule{0.98\linewidth}{1pt}}
\begin{tabular}[t]{rrr}
{\bf Total LOC}: 13842& \hspace{3em}
{\bf Total VCs}: 5815& \hspace{3em}
{\bf Total Time (s)}: 387.94
\end{tabular}

  \caption{
    Summary of evaluation results for our verification tool,
    featuring lines of code,
    number of verification conditions (VCs),
    and type-checking times (including generation of VCs and
    checking of VCs by Inox).
  }
  \label{fig:evaluation}

\end{figure}


To illustrate the diversity of benchmarks, we note that 
\bname{InsertionSort}, \bname{QuickSorts}, \bname{MergeSorts} and
\bname{StableSort} 
feature various implementations of the sorting algorithms with the typical properties
shown.
\bname{ListMonad} and \bname{OptionMonad}
show that the monadic laws hold for the
$\PList{\tt X}$ and {\tt Option[X]} types.
The \bname{List} and
\bname{ListWithSize} benchmarks feature a collection of
common higher-order functions on lists such as \lstinline!map!,
\lstinline!filter!, \lstinline!forall!, etc., as well as many
properties shown about the implementations.
In the \bname{GodelNumbering} benchmark, we prove that the pairing function
$2^x(2y + 1) - 1$ is a bijection between natural numbers and pairs of natural
numbers using a series of lemmas about linear and non-linear
arithmetic. Similar non-linearity is featured in the
\bname{MoreExtendedEuclidGCD} benchmark where we show that an implementation
of the extended Euclid's algorithm indeed computes the greatest common divisor
and the coefficients of Bézout's identity. Benchmarks also include persistent data structures
\cite{OkasakiFunDS},
ConcTree and ConcRope \cite{conq}, explicit state model checker ReachabilityChecker,
small interpreters, dynamic programming algorithms such
as Viterbi and Knapsack, and benchmarks solving assignments
from Scala MOOCs (see \url{https://courseware.epfl.ch/} or
\url{https://www.coursera.org/specializations/scala}).


\section{Related Work}

In this paper, we tackle program verification using a type-theoretic approach.
Other techniques, such as symbolic execution for verification (see
e.g.~\cite{DBLP:journals/jfp/NguyenTH17,DBLP:conf/pldi/HallahanXBJP19}) or using
term rewriting systems (see
e.g.~\cite{AutomatedTerminationAproveGiesl2004,Giesl2006,Giesl2011ATP}) can be used
instead or in complement to our approach.
The use of type checking in verification appears in program verifiers such as
F*~\cite{Swamy2013,DBLP:conf/popl/AhmanHMMPPRS17,DBLP:conf/popl/SwamyHKRDFBFSKZ16} and
Liquid Haskell~\cite{VazouRondonJhala13AbstractRefinementTypes,LiquidHaskell,DBLP:journals/pacmpl/VazouTCSNWJ18}.
Like Dafny~\cite{dafny} and Stainless, these systems rely on SMT solvers to automatically discharge
verification conditions.

F* is a dependently-typed programming language which
supports a rich set of effects (such as divergence, mutation, etc.). To the best of our knowledge
F* does not support proving termination of functions operating on
infinite data structures such as streams.
Stainless has a desugaring pass that handles local state and
certain forms of unique references. Furthermore, for
semantic modeling of global state, a user can use monad
syntax (Scala's \lstinline|for| comprehensions) explicitly in the
surface language, which Scala compiler desugars into
higher-order functions that our system can handle. 
Overall, F* supports a more expressive class of effects than what Stainless
currently handles.
The goal of Liquid Haskell~\cite{DBLP:journals/pacmpl/VazouTCSNWJ18} is to add refinement types to Haskell, a call-by-need language (while we focus here on call-by-value languages). The system supports a sound
verification procedure which relies on decidable theories of SMT.
The supported language imposes certain restrictions on refinement occurrences
and the metatheory does not feature recursive types (although they are supported
by the implementation).
Similarly to our body-visible recursion, Liquid Haskell allows the type checker to
access the body of recursive functions (for recursive calls) while type-checking
the function itself using a technique called \emph{refinement reflection}~\cite{DBLP:journals/pacmpl/VazouTCSNWJ18}.
Dafny verifier~\cite{dafny} supports imperative and object-oriented as well as functional programming,
including inductive and coinductive types,
recursion using ordinals, quantifiers. We are not aware of metatheory to
justify the soundness of Dafny.
In contrast, we provide a mechanized proof of soundness for System FR, which
increases confidence in VC generation soundness. This is important because
it is easy to construct paradoxes when
combining features  such as impredicativity and
contravariant recursion. Being based on a reducibility relation, our proof can be used
as a basis to add new language features in a compositional way. Several other
systems have been recently formalized, such as MetaCoq~\cite{sozeau:hal-02167423}.


Proof assistants with great expressive power include
Isabelle~\cite{Isa,DBLP:conf/esop/BlanchetteBL0T17},
Coq~\cite{DBLP:series/txtcs/BertotC04,DBLP:journals/jfrea/Barras10,DBLP:conf/lics/Sacchini13},
Idris~\cite{DBLP:conf/plpv/Brady13},
Agda~\cite{norell2007towards,DBLP:conf/itp/000110},
PML$_2$~\cite{DBLP:phd/hal/Lepigre17}, Lean~\cite{DBLP:conf/isaim/Moura16} and
Zombie~\cite{DBLP:conf/popl/CasinghinoSW14}. When coinduction is supported in
these systems (which is not always the case), it is typically through
special constructs for coinductive types and their corresponding
cofixpoint operators, dual to inductive types and their fixpoint operators. In
System FR, we instead treat induction and coinduction in a uniform way (see
e.g.~Section~\ref{sec:running-example2-streams}). Our type system features a
single kind of recursive types which allow uniform definition of inductive,
coinductive and mixed recursive types. Our operational semantics further rely on
the {\tt fix} operator which roughly corresponds to general recursion in
mainstream programming call-by-value languages. An alternative approach to
uniform handling of recursion and corecursion is given in
\cite{DBLP:conf/types/GianantonioM02}. We were able to encode the first example
from this work (in
\href{https://github.com/epfl-lara/stainless/blob/1d324615ab9a328180d7e08a51051fcad200a414/frontends/benchmarks/dotty-specific/typechecker/Streams.scala#L42}{our Streams benchmark}) which generates a stream and whose termination proof
requires a combination of inductive and coinductive reasoning. Using our system,
we expressed that using a simple lexicographic measure combining the index of
the type of the produced stream with the inductive measure. In a private
communication, Andreas Abel showed us that Agda can also handle this example.

Our system follows Nuprl~\cite{DBLP:books/daglib/0068834} (and other
computational type theories) style of starting out with an untyped calculus and
then introducing various types to classify untyped terms based on their
behaviors. Nuprl supports a very expressive type system, which covers the types
we present in this paper, except impredicative polymorphism. In a large
development, Nuprl's metatheory
has been formalized in Coq~\cite{anand2014towards}. One part of that formalization
relies on the use of the Coq axiom of functional choice {\tt FunctionalChoice\_on},
which gives a function $f$ from $A$ to $B$ when a formula
of the form $\forall a: A.\ \exists b: B.\ \phi(a,b)$ holds. Nuprl does not use SMT solvers for
automation, but relies instead on built-in and user-defined tactics similarly to Coq.

Our reducibility definition for recursive type is inspired from
step-indexed logical relations~\cite{DBLP:conf/esop/Ahmed06}. The main
difference is that the indices in step-indexed logical relations do not appear
at the level of types, but at the level of the logical relation that gives
meaning to types. In System FR we internalize the indices at the
level of recursive types in order to give more expressive power to the users,
and let them specify decreasing measures for recursive functions that manipulate
infinite data structures.
This treatment of recursive types is similar to the
TORES~\cite{DBLP:conf/rta/Jacob-RaoPT18} type system, where recursive types can
be indexed by an arbitrary index language. We only support recursive types
indexed by natural numbers. TORES provides a decision procedure for a rich type
system with inductive types, coinductive types, and indexed recursive
types, yet its metatheory does not handle polymorphism nor refinement types.

Our termination criterion (rule for {\tt fix}) is inspired by
type-based termination~\cite{DBLP:journals/mscs/BartheFGPU04,DBLP:conf/icfp/AbelP13,DBLP:journals/corr/abs-1202-3496,DBLP:journals/ita/Abel04,DBLP:journals/lmcs/Abel08,DBLP:phd/de/Abel2007}.
Such work also typically uses two different kinds of recursion, one for
induction and one for coinduction. In type-based termination, instead of
requiring that a measure on the arguments of recursive calls decreases, we
require that recursive functions are called at a
\emph{type} which is strictly smaller than the type of the caller. Our {\tt fix}
operator produces a term with a forall type, which is similar to the implicit
function type of the implicit calculus of
constructions~\cite{DBLP:conf/tlca/Miquel01}. Our termination measure on types
can then be understood as a measure on the implicit argument of a function with
an implicit function type. Earlier work on type-based termination checking
includes simpler type systems such as DML \cite{DependentTypesXi2001}.

A proof of concept of extension of Dotty compiler by a restricted set
of (mostly numerical)
refinement predicates was presented in \cite{SchmidKuncak16CheckingPredicate}, without soundness proof.
Developments specific to Scala include dependent object
types \cite{Amin16PhD} which focus on path dependency instead of
predicate refinement and, by design, admit possibly
non-terminating programs. System FR does not attempt to support
records, let alone path dependency. Stainless relies on Scala compiler front ends to
obtain a symbol-resolved syntax tree. For example, implicit
parameters in Scala \cite{DBLP:journals/pacmpl/OderskyBLBMS18} become ordinary explicit
parameters by the time Stainless processes them. Conversely, many of the types we define
in System FR do not have their counterparts in Scala.
An upcoming PhD thesis of Nicolas Voirol
\cite{Voirol19PhD} presents a system related to System FR that more closely
follows the expressive power of Inox solver on which Stainless relies.

\section{Conclusion}

We have presented System FR, a formalized type system and a
bidirectional type checking algorithm that can serve as a basis
of a verifier for higher-order functional programs. 
We were able to verify correctness and safety of a wide
range of benchmarks amounting to 14k lines of code and
proofs.
Our formalization suggests that
lazy data structures and
non-covariant recursion are tractable and
that explicit indices required in the general framework can
be eliminated in commonly occurring cases.
Our type system incorporates $\Pi$ and $\Sigma$ types, yet
it also supports their variants ($\forall$ and $\exists$) that
correspond to infinite intersections and unions. Along with support for
singleton types and refinements, we obtained a rich
framework to approximate program semantics, which can also help
further type, invariant, and measure inference algorithms.
Our experience confirms an advantage of the semantic-based
soundness proof: once we adopt an approach of interpreting
types as sets of terms, we are less dependent on a
particular choice of syntactic rules; we can introduce new
classes of types and new rules for verification condition
generation and solving, as long as we can justify them
semantically.
When viewing types in our system as propositions, we obtain
an expressive quantified logic. We have proven the
soundness of this logic in Coq using
countable interpretation of ground System FR types.
The success in verification of our benchmarks suggests that the rules of System FR work
well for many properties of functional programs. 
Furthermore, refinement by type along with intersections and unions
(Section~\ref{sec:ext}) allows System FR
to describe types whose interpretations are undecidable countable sets from higher levels of
arithmetical hierarchy.

%
%
%






\subsection*{ACKNOWLEDGMENTS}
Work supported in part by Swiss National Science Foundation
project 200021\_175676 ``Scaling Predicate Types''. We thank
Ravichandhran Kandhadai Madhavan for numerous discussions
about termination for higher-order functional programs;
Romain Ruetschi for his valuable contributions to Stainless
implementation; Paolo G. Giarrusso for interesting
discussions and very relevant pointers on type theory;
Andreas Abel for explaining us how to mix induction and
co-induction in Agda. Thanks to Yoan Géran who implemented
another prototype of our algorithm, found formalization
discrepancies and proposed improvements to the typing rules.

\bibliography{managed,more,ravi_thesisbib,icfp-2019}

\newpage

\appendix

\section{Lexicographic Orderings}
\label{app:lexicographic}


\begin{figure}[htbp]
\begin{tabular}{l|l}
\begin{minipage}[c]{0.37\textwidth}
\begin{lstlisting}
def f(x: T): R {
  decreases($m_1$(x), $m_2$(x))
  $E$
}
\end{lstlisting}
\end{minipage}
&
\begin{minipage}[c]{0.6\textwidth}
\begin{lstlisting} 
def f(x: T): R = {
  decreases($m_1$(x))
  def g(y: T): R = {
    require($m_1$(y) == $m_1$(x))
    decreases($m_2$(y))
    $E$[x:= y, f $z$ := if ($m_1$($z$) < $m_1$(x)) f($z$) else g($z$)]
  }   
  g(x)  
}
\end{lstlisting}
\end{minipage}
\\ \hline
\begin{minipage}[c]{0.37\textwidth}   
\begin{lstlisting}
def A(m: Nat, n: Nat): Nat = {
  decreases(m, n)
  if (m = 0) n+1
  else if (n = 0) A(m-1, 1)
  else A(m-1, A(m, n-1))
}
\end{lstlisting}
  \end{minipage} 
  & 
  \begin{minipage}[c]{0.6\textwidth}
\begin{lstlisting}
def A(m: Nat, n: Nat): Nat = {
  decreases(m)
  def Ag(my: Nat, ny:Nat): Nat = {
    require(my == m)
    decreases(ny)
    if (my = 0) ny+1
    else if (ny = 0) A(my-1, 1)
    else A(my-1, Ag(my, ny-1))
  }
  Ag(m, n) 
}
\end{lstlisting}
  \end{minipage}  
\end{tabular}
\caption{Encoding of Lexicographic Orderings through Mutual Recursion.
Top row shows a general scheme; bottom shows Ackerman's function as an example.
The left side is source code, the right is the encoding.\label{fig:lexicographic}}
\end{figure}

Functions whose termination arguments require lexicographic orderings can be
encoded by using two levels of recursions, which is a known technique that shows
expressive power of System T \cite[Section 7.3.2]{ProofsAndTypesBook}.
The right-hand-side uses (twice) the syntactic
sugar defined in Figure~\ref{fig:sugar} and described in the previous section.
The outermost recursion allows recursive calls whenever the first measure
decreases, while the innermost one is used when the first measure stays
the same and the second measure decreases.

For the encoding, we assume that in the body of function {\tt f} (i.e.~in the
expression $E$), {\tt f} is always applied to some argument. The notation
\[
    \texttt{$E$[f $z$ := if ($m_1$($z$) < $m_1$(x)) f($z$) else g($z$)]}
\]
represents the term where that every application of the form {\tt f $z$} in $E$
is replaced by the corresponding \tite{}{}{} expression. This expression checks
at runtime which measure decreases, and decides to call the outermost or
innermost recursion. When the user knows which measure decreases for a given
recursive call, the \tite{}{}{} expression can be optimized away by directly
calling the appropriate branch. We give as an example the Ackermann function
(see Figure~\ref{fig:lexicographic}), which uses the lexicographic ordering of its
argument for ensuring termination.


\section{Strict Positivity of a Type Variable in a Type}
\label{app:polarity}



\begin{figure}
\begin{flalign*}
  & \tspos{\alpha}{\tau} \define {\tt true} \textrm{ if $\alpha \not\in \fv{\tau}$}\\
  & \tspos{\alpha}{\alpha} \define {\tt true} \\
  & \tspos{\alpha}{\trefine{x}{\tau}{b}} \define \tspos{\alpha}{\tau} \\
  & \tspos{\alpha}{\tarrow{x}{\tau_1}{\tau_2}} \define \alpha \not\in \fv{\tau_1} \land \tspos{\alpha}{\tau_2} \\
  & \tspos{\alpha}{\tforall{x}{\tau_1}{\tau_2}} \define \alpha \not\in \fv{\tau_1} \land \tspos{\alpha}{\tau_2} \\
  & \tspos{\alpha}{\tpoly{\beta}{\tau}} \define \tspos{\alpha}{\tau} \\
  & \tspos{\alpha}{\tsum{\tau_1}{\tau_2}} \define \tspos{\alpha}{\tau_1} \land \tspos{\alpha}{\tau_2} \\
  & \tspos{\alpha}{\tprod{x}{\tau_1}{\tau_2}} \define \tspos{\alpha}{\tau_1} \land \tspos{\alpha}{\tau_2} \\
  & \tspos{\alpha}{\indexedtype{n}{\beta}{\tau}} \define
      \tspos{\alpha}{\tau} \land
      (\alpha \notin \fv{\tau}\,\lor \tspos{\beta}{\tau})
\end{flalign*}

\caption{Definition of strict positivity for a type variable $\alpha$ in a type.}
\label{fig:spos}
\end{figure}

A type variable $\alpha$ is said to be \emph{strictly positive} (see
Figure~\ref{fig:spos}) in a type $\tau$, if it only appears to the
right-hand-sides of $\Pi$ and $\forall$ types. This restriction is used in the
additional typing rules that are given in Section~\ref{sec:generalization}.



\section{Erasure of Type Annotations in Terms}
\label{app:erase}

In Section~\ref{sec:annotated-terms}, we use the notation $\erase{t}$ to refer
to the \emph{erasure} of an annotated term $t$. The precise definition is
given in Figure~\ref{fig:erase}. Type annotations are used to guide our
type-checking algorithm but play no role in the reducibility definition or in
the operational semantics, which talk about erased terms with no annotation.


\begin{figure}[htb]

\begin{alignat*}{2}
& \erase{x} &&\define x \\
& \erase{\tu} &&\define \tu \\
& \tlambda{x}{\tau}{t} &&\define \tulambda{x}{\erase{t}} \\
& \erase{\tapp{t_1}{t_2}} &&\define \tapp{\erase{t_1}}{\erase{t_2}} \\
& \erase{(t_1,t_2)} &&\define (\erase{t_1}, \erase{t_2}) \\
& \erase{\tprojl{t}} &&\define \tprojl{\erase{t}} \\
& \erase{\tprojr{t}} &&\define \tprojr{\erase{t}} \\
& \erase{\taleft{\tau}{t}} &&\define \tuleft{\erase{t}} \\
& \erase{\taright{\tau}{t}} &&\define \turight{\erase{t}} \\
& \erase{\tsummatch{t_1}{x \Rightarrow t_2}{x \Rightarrow t_3}} &&\define
  \tsummatch{\erase{t_1}}{x \Rightarrow \erase{t_2}}{x \Rightarrow \erase{t_3}} \\
& \erase{\ttrue} &&\define \ttrue \\
& \erase{\tfalse} &&\define \tfalse \\
& \erase{\tite{t_1}{t_2}{t_3}} &&\define \tite{\erase{t_1}}{\erase{t_2}}{\erase{t_3}}\\
& \erase{\tzero} &&\define \tzero \\
& \erase{\tsucc{t}} &&\define \tsucc{\erase{t}} \\
& \erase{\tmatch{t_n}{t_0}{n \Rightarrow t_s}} &&\define
  \tmatch{\erase{t_n}}{\erase{t_0}}{n \Rightarrow \erase{t_s}} \\
& \erase{\trec{x \Rightarrow \tau}{t_n}{t_0}{(n,y) \Rightarrow t_s}} &&\define
  \turec{\erase{t_n}}{\erase{t_0}}{(n,y) \Rightarrow \erase{t_s}} \\
& \erase{\tfix{n \Rightarrow \tau}{(n,y) \Rightarrow t}} && \define
  \tufix{y \Rightarrow \erase{t}} \textrm{\hspace{3em}(assuming $n \not\in \fv{\erase{t}}$)}
\\
& \erase{\tinstforall{t_1}{t_2}} && \define \erase{t_1} \\
& \erase{\tfold{\tau}{t}} && \define \tufold{\erase{t}} \\
& \erase{\tunfoldin{t_1}{x \RA t_2}} && \define \tunfoldin{\erase{t_1}}{x \RA \erase{t_2}} \\
& \erase{\tunfoldposin{t_1}{x \RA t_2}} && \define \tunfoldin{\erase{t_1}}{x \RA \erase{t_2}} \\
& \erase{\tabs{\alpha}{t}} && \define \tuabs{\erase{t}} \\
& \erase{\tinst{t}{\tau}} && \define \tuinst{\erase{t}} \\
& \erase{\terr{\tau}} && \define \tuerr \\
& \erase{\trefl{t_1}{t_2}} && \define \tu \\
& \erase{\tulet{x}{t_1}{t_2}} && \define \tulet{x}{\erase{t_1}}{\erase{t_2}} \\
& \erase{\tsize{t}} && \define \tsize{\erase{t}} \\
\end{alignat*}

\caption{Erasing type annotations.}
\label{fig:erase}

\end{figure}

\section{Proof of Lemma~\ref{lem:reclists}}

\reclists*

\begin{proof}

Given a list of the form $\tufold{\turight{a_1, \dots\tufold{\turight{a_k, \tufold{\tuleft{}}}}\dots}}$,
we say that $k$ is its \emph{size}.

($\Leftarrow$) We prove by induction on (the size of) $n \in \redvalues{\tnat}{}$
that $\redvalues{\IPList{\tt X}{n}}{}$ contains all finite lists. Then we can conclude
that
$\redvalues{\List}{}$
contains all finite lists.
\begin{itemize}
\item $(n = \tzero)$. By definition, $\redvalues{\IPList{\tt X}{0}}{}$ contains all
    values of the form $\tufold{v}$, and therefore all finite lists.
\item $(n = \tsucc{n'})$. The induction hypothesis tells us that
    $\redvalues{\IPList{\tt X}{n'}}{}$ contains all finite lists. Consider a list
    $v \in \Val$, $k \geq 0$ and $a_1,\dots,a_k \in \redvalues{\tt X}{}$ such
    that:
\[
    v = \tufold{\turight{a_1, \dots\tufold{\turight{a_k, \tufold{\tuleft{}}}}\dots}}.
\]

We distinguish two cases, if $k = 0$, i.e.~$v = \tufold{\tuleft{}}$ is the empty
list, then we conclude directly by definition of $\redvalues{\IPList{\tt
X}{n}}{}$ that $v \in \redvalues{\IPList{\tt X}{n}}{}$.

If $k > 0$, then $v = \tufold{\turight{a_1, \dots\tufold{\turight{a_k,
\tufold{\tuleft{}}}}\dots}}$. By definition of $\redvalues{\IPList{\tt X}{n}}{}$, we
know that $v \in \redvalues{\IPList{\tt X}{n}}{}$ if and only if the tail of $v$ belongs
to $\redvalues{\IPList{\tt X}{n'}}{}$, i.e.~$\tufold{\turight{a_2,
\dots\tufold{\turight{a_k,\tufold{\tuleft{}}}}\dots}} \in
\redvalues{\IPList{\tt X}{n'}}{}$, which is true by induction hypothesis.
\end{itemize}

$(\Rightarrow$) In a first step, we prove by induction on (the size of) $n \in
\redvalues{\tnat}{}$ the following statement ($\ast$): if $v \in
\redvalues{\IPList{\tt X}{n}}{}$, then either $v$ is a list of size at most $n-1$, or
there exists $a_1,\dots,a_n \in \redvalues{\tt X}{}$ and $b \in \Val$ such
that $v = \tufold{\turight{a_1, \dots\tufold{\turight{a_n, b}}\dots}}$.

We can then use this fact to prove that if $v \in \redvalues{\PList{\tt X}}{}$,
then $v$ is a finite list, as follows. By definition, we know that, for every $n
\in \redvalues{\tnat}{}$, $v \in \redvalues{\IPList{\tt X}{n}}{}$. The term $v$
is represented by a finite syntax tree, and therefore there exists an $n$
(e.g.~the size of the syntax tree of $v$ plus one) such that $v$ cannot be of
the form $v = \tufold{\turight{a_1, \dots\tufold{\turight{a_n, b}}\dots}}$. By
using ($\ast$), it follows that $v$ must be a list of size at most $n-1$.

Let us now proceed to the proof of ($\ast$):
\begin{itemize}
\item $(n = \tzero)$ The statement holds since $\IPList{\tt X}{0}$ represents the values
of the form $\tufold{v}$ where $v \in \Val$, and we can choose $b = v$.
\item $(n = \tsucc{n'})$ By definition of $\redvalues{\IPList{\tt X}{n}}{}$, we
 know that either $v = \tufold{\tuleft{}}$ is the empty list, or there exists $a
 \in \redvalues{\tt X}{}$ and $v' \in \redvalues{\IPList{\tt X}{n'}}{}$ such that
 $v = \tufold{\turight{a,v'}}$. By induction hypothesis, we know that $v'$ is
 either a list of size at most $n'-1$ or there exists $a_1,\dots,a_{n'} \in
 \redvalues{\tt X}{}$ and $b \in \Val$ such that $v' = \tufold{\turight{a_1,
 \dots\tufold{\turight{a_{n'}, b}}\dots}}$.

 We therefore conclude that $v$ is either a list of size at most $n-1$, or that:
 \[
     v = \tufold{\turight{a,\tufold{\turight{a_1, \dots\tufold{\turight{a_{n'}, b}}\dots}}}}.
  \]
\end{itemize}

\end{proof}


\section{Operations on Natural Numbers\label{app:lessthan}}


\begin{figure}[htb]

\begin{minipage}{0.49\textwidth}
    \begin{tabular}{c}
    \begin{lstlisting}
lessThan $\triangleq$ $\lambda$a: $\tnat$.
    rec[_ $\Rightarrow$ $\tnat\rightarrow\tbool$](
        a,
        $\lambda$x: $\tnat$. ${\tt match}$(x, $\tfalse$, _ $\Rightarrow$ $\ttrue$),
        (_, y) $\Rightarrow$ $\lambda$x: $\tnat$. ${\tt match}$(x, $\tfalse$, n $\Rightarrow$ y n)
    )
    \end{lstlisting}
    \end{tabular}
\end{minipage}
\begin{minipage}{0.49\textwidth}
    \begin{tabular}{c}
    \begin{lstlisting}
lessEqual $\triangleq$ $\lambda$a: $\tnat$.
    rec[_ $\Rightarrow$ $\tnat\rightarrow\tbool$](
        a,
        $\lambda$x: $\tnat$. ${\tt match}$(x, $\ttrue$, _ $\Rightarrow$ $\ttrue$),
        (_, y) $\Rightarrow$ $\lambda$x: $\tnat$. ${\tt match}$(x, $\tfalse$, n $\Rightarrow$ y n)
    )
    \end{lstlisting}
    \end{tabular}
\end{minipage}

\begin{tabular}{c}
\begin{lstlisting}
equalNat $\triangleq$ $\lambda$a: $\tnat$.
    rec[_ $\Rightarrow$ $\tnat\rightarrow\tbool$](
        a,
        $\lambda$x: $\tnat$. ${\tt match}$(x, $\ttrue$, _ $\Rightarrow$ $\tfalse$),
        (_, y) $\Rightarrow$ $\lambda$x: $\tnat$. ${\tt match}$(x, $\tfalse$, n $\Rightarrow$ y n)
    )
\end{lstlisting}
\end{tabular}

\caption{For values $a,b \in \redvalues{\tnat}{}$, {\tt lessThan $a$ $b$}
(resp.~{\tt lessEqual}, resp.~{\tt equalNat}) returns $\ttrue$ if the natural
number represented by $a$ is strictly less (resp.~less or equal, resp.~equal)
than the one represented by $b$.}
\label{fig:natops}

\end{figure}

Figure~\ref{fig:natops} shows the definitions of {\tt lessThan}, {\tt
lessEqual}, and {\tt equalNat} that implement comparison and equality on the
$\tnat$ type.


\end{document}